\documentclass[10pt]{article}
\usepackage{epsfig,graphics,graphicx,subfigure,color}
\usepackage{amsfonts,amssymb,amsthm,amsmath,amsbsy,mathtools,cuted,kotex,subeqnarray,ocgx}

\newcommand{\interior}[1]{  {\kern0pt#1}^{\mathrm{o}}}
\newcommand{\norm}[1]{\left\lVert#1\right\rVert}

\newtheorem{theorem}{Theorem}
\newtheorem{corollary}[theorem]{Corollary}
\newtheorem{lem}[theorem]{Lemma}

\def\Bbb R{{\rm \bf R}}

\def\gathered{\begin{array}{c}}
\def\endgathered{\end{array}}
\def\text{\mbox}

\begin{document}
\title  {On the Uniqueness of Solutions in GPS Source Localization: Distance and Squared-Distance Minimization under Limited Measurements in Two and Three Dimensions}
\author{Kiwoon KWON\footnote{ Department of Mathematics, Dongguk University-Seoul, 04620 Seoul, South Korea.
e-mail address:kwkwon@dongguk.edu}
}
\maketitle

\begin{abstract}
The source localization problem, fundamental to applications like GPS, is typically approached as a minimization problem in the presence of various types of noise. Ensuring the uniqueness of solutions in GPS technology is vital for the reliability and accuracy of applications, from everyday navigation to critical military operations. In this paper, we examine two key minimization problems: one focused on distance error and the other on squared distance error. We explore these problems in both three-dimensional space, the standard scenario, and in two-dimensional space as a simplified case. Furthermore, we discuss the number of possible source solutions when the number of measurements is fewer than three.



\noindent
AMS: 
\end{abstract}

\section{Introduction}
The uniqueness of the source localization problem in Global Positioning System (GPS) technology is a fundamental aspect that ensures the reliability and accuracy of various applications, ranging from everyday navigation to critical military operations. This paper explores key issues related to the uniqueness of source localization and its implications for the stability and performance of GPS systems.

A central challenge in source localization is accurately and uniquely determining the position of a signal emitter despite the presence of noise, multipath effects, and other environmental factors. Multipath propagation, where GPS signals reflect off surfaces such as buildings and terrain, can introduce significant errors and ambiguities in position estimation. This issue is particularly severe in urban and densely built environments, leading to potential misidentification of the true signal source.

Another key factor is satellite geometry, the geometric configuration of the satellites in view. Poor satellite geometry can lead to weak geometric dilution of precision (GDOP), making it difficult to achieve unique and precise localization, especially in environments with limited satellite visibility, such as deep urban canyons or densely forested areas. Furthermore, precise time synchronization between the GPS satellites and the receiver is critical, as discrepancies in clock synchronization can introduce errors that compromise the uniqueness and accuracy of the localization process. Synchronizing time accurately across multiple satellites and receivers is a technical challenge that requires advanced algorithms and error correction techniques.

Additionally, the presence of multiple potential signal sources can complicate the localization process. Distinguishing between these sources and accurately identifying the correct one is essential for maintaining the stability and reliability of the GPS system, necessitating sophisticated signal processing techniques and robust algorithms capable of handling complex scenarios with overlapping signals.

In this paper, we investigate the conditions under which the proposed minimization problems yield unique solutions, particularly when the number of measurements is fewer than three. This aspect is crucial, as fewer than three measurements can make the problem significantly more challenging, leading to potential ambiguities and non-uniqueness in the solution. Understanding how this limitation impacts the problem and identifying strategies to address it are key contributions of this work.

We formulate the source localization problem as follows: Let 
 $Z_i \in \mathbb R^n$ (with $n=2$ or $n=3$ represent the location of the 
 $i$-th sensor ($i=1,\cdots,I$ for $I=1,2,3$), and let $d_i\ge0$ denote the measured distance between the source and the $i$-th sensor. 
 The goal is to determine $X$,  the location that minimizes the sum of the additive noise, expressed as:
 \begin{equation}\label{eq:main}
                    X = \mathrm{argmin}_{\tilde W} O(\tilde W) = \{W\in \mathbb R^2 \mbox{ or } \mathbb R^3 | O(W) = \min_{\tilde W} O(\tilde W) \}, \;   O(W)= \sum_{j=1}^I|\psi(\norm{W-Z_i}) - \psi(d_i)|
 \end{equation}
where $\psi$ is a non-negative, continuous, increasing function defined on $[0,\infty)$, vanishing only at zero. In this paper, we specifically consider the cases 
$\psi(y)=y$ and $\psi(y)=y^2$, which is called in this paper `distnace error' and `squared distance error', respectively. 

Since $X$ is understood as a set rather than a single point, we define 
$|X|$ as the number of elements in the set $X$. The objective function 
$O$ is nonnegative, continuous, and satisfies:
$$\lim_{||W||\rightarrow \infty}O(W)=\infty$$ 
uniformly in all directions in $\mathbb R^2$ or $\mathbb R^3$.
Therefore, by similar arguments as in \cite{KL1, KL2, KL2Cor}, a solution 
$X$ exists for the equation above.

Let $D_i, i=1,\cdots, I$ be the measurement closed ball centered at 
$Z_i$  with radius $d_i$, and let $S_i$ be its boundary, $D_i^o = D_i\setminus S_i$, and
$$
S=\cup_{j=1}^I S_j.     
$$
In the following sections, we investigate the number of solutions and their locations when the number of measurements is one (Section 2), two (Section 3), and three (Section 4).

\section{One measuremt (I=1)}
Let us first consider $I=1$. The solution $X$ is a circle in two dimenstions and a sphere in three dimensions, both centered at $\mathbb O$ with a radius $d_1$. 
In other words,
$$ X=S_1, \mbox{ and } |X|=1. $$
However, the behavior of the level set differs between the cases $\psi(y)=y$ and $\psi(y)=y^2$ as illustrated in Fig. \ref{fig:Meas1}. In the distance error case ($\psi(y)=y$), the Euclidean distance between the level sets remains uniform. 
In contrast, for the squared distance error case ($\psi(y)=y^2$), the Euclidean distance decreases as the location moves furher away from the measurment circle (which is depicted as the unit circle in Fig. \ref{fig:Meas1}). Consequently, the 
minimization problem for squared distance error exhibits faster local convergence compared to the problem with distance error. This behavior is consistent in both two-dimensional and three-dimensional spaces. 

In figures \ref{fig:Meas1} (g) and \ref{fig:Meas1} (h), three dimensional gradient
vectors are plotted. The gradient directions point outward if the starting points are outside the unit sphere (the 0-level set)  and inward if the starting points are inside the unit sphere. The length of the gradient vectors remains constant 
in the distance error case in (g), whereas in the squared distance error case (h), the length is propotional to the distance between the starting point and the origin.

 \begin{figure}
\begin{minipage}[t]{17cm}
\centerline{\epsfig{file=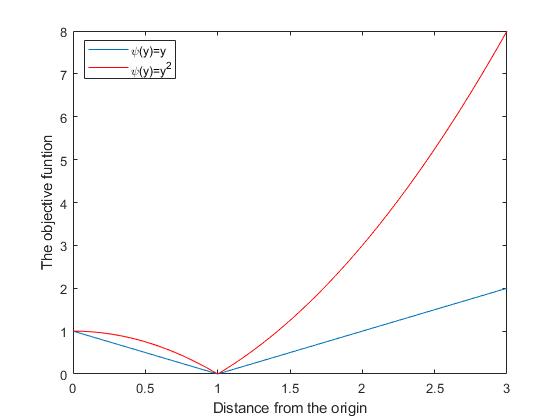, height=3.5cm,width=6cm,clip=1cm}}
\end{minipage}
\begin{center}
(a)
\end{center}
\begin{minipage}[t]{8cm}
\centerline{\epsfig{file=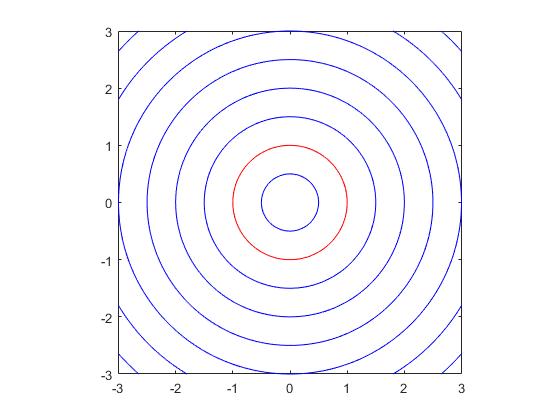, height=3.5cm,width=5.5cm,clip=1cm}}
\end{minipage}
\begin{minipage}[t]{8cm}
\centerline{\epsfig{file=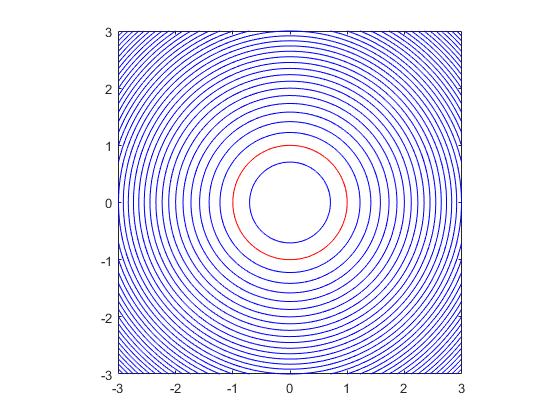, height=3.5cm,width=5.5cm,clip=1cm}}
\end{minipage}
\begin{center}
(b)\qquad\quad\qquad\qquad\qquad\qquad\qquad\qquad\qquad\qquad\qquad(c)
\end{center}
\begin{minipage}[t]{8cm}
\centerline{\epsfig{file=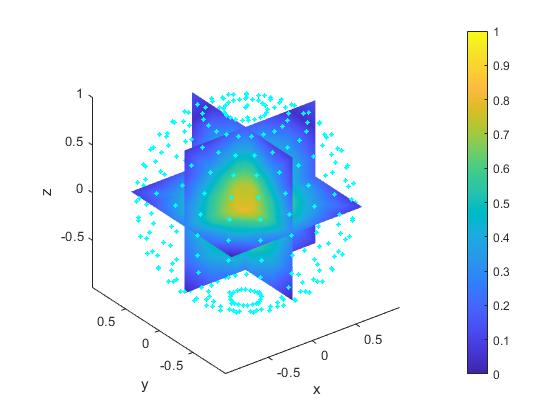, height=3.5cm,width=5.5cm,clip=1cm}}
\end{minipage}
\begin{minipage}[t]{8cm}
\centerline{\epsfig{file=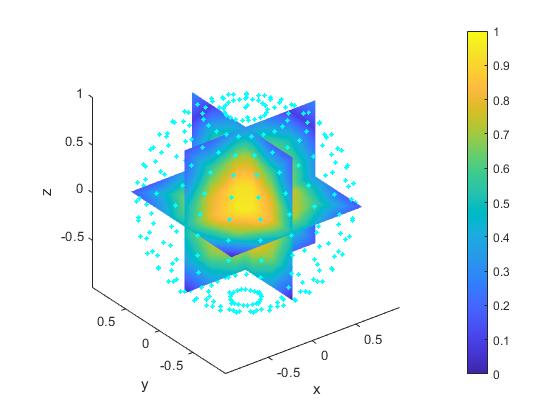, height=3.5cm,width=5.5cm,clip=1cm}}
\end{minipage}
\begin{center}
(d)\qquad\quad\qquad\qquad\qquad\qquad\qquad\qquad\qquad\qquad\qquad(e)
\end{center}
\begin{minipage}[t]{8cm}
\centerline{\epsfig{file=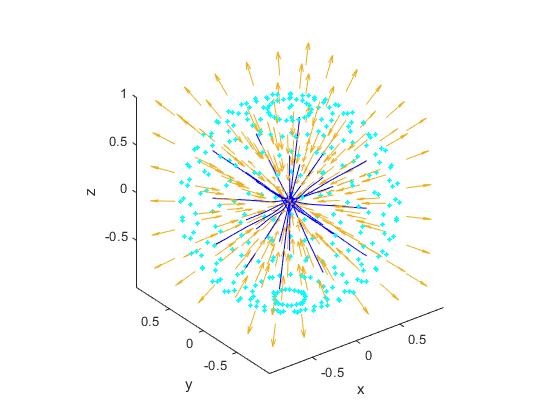, height=3.5cm,width=5.5cm,clip=1cm}}
\end{minipage}
\begin{minipage}[t]{8cm}
\centerline{\epsfig{file=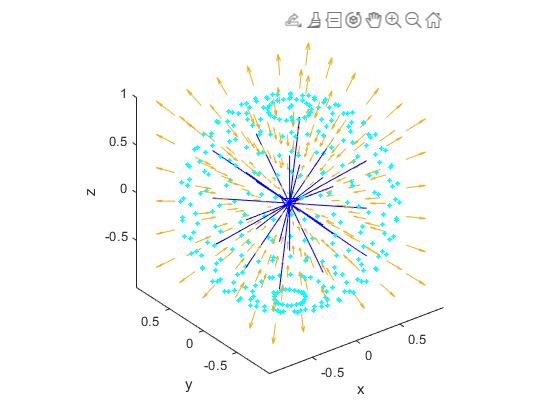, height=3.5cm,width=5.5cm,clip=1cm}}
\end{minipage}
\begin{center}
(f)\qquad\quad\qquad\qquad\qquad\qquad\qquad\qquad\qquad\qquad\qquad(g)
\end{center}
\caption{(a) The relation between the distance to $Z_1$ and the objective function is shown for both the distance error (blue line) and the squared distance error(red line).  (b,c) The level sets of the objective function are depicted for $Z_1=\mathbb O$ and $d_1=1$ in the $xy$-plane.The red line is the level set $0$ and the blue lines are level sets for $0.5, 1, 1.5,$ and so on.  A three dimensional sphere is plotted with blue dots. In (d) and (e), the values of level sets are displayed in $xy-, yz-, zx-$ planses. 
In (f) and (g), the streamlines
and the gradient vectors are shown. Specifically, (b),(d), and (e) corresponds to the distance error case, while (c),(d), and (g) represent the squared distance error case. }
\label{fig:Meas1} 
\end{figure}

\section{Two measuremts (I=2)}
Let us decompose $\mathbb R^n, n=2,3$ as four regions:
$$ D_1\setminus D_2, D_2 \setminus D_1, D_1 \cap D_2, (D_1\cup D_2)^c. $$
Define $S_{12}=S_1 \cap S_2$ and $N_{\pm 1} = Z_1 \pm d_1 \frac{Z_2 - Z_1}{||Z_2-Z_1||}$. 
Similarily, $N_2, N_{-2}$ can be defined. 
If $S_{12}\neq\phi$, then  $|S_{12}|=1,2$ for $n=2$ and $|S_{12}|=1,\infty$ for $n=3$.
Without loss of generality, assume $d_1 \ge d_2$.

The singular point for the objective function is denoted by $Y_0$ and its definition of $Y_0$ is different 
for distance and squared distace errors as follows:
$$ Y_0 = \left\{
\begin{array}{ccc}
               \frac{Z_1 + Z_2}2                                              &\mbox{ for }&   \psi(y) = y^2, \\
               \overline{Z_1 Z_2} \setminus (D_1 \cup D_2)  &\mbox{ for }&   \psi(y)=y. \\
\end{array}\right.
$$

\subsection{The squared distance error case  ($\psi(y)=y^2$)}
Note that $Y_0= \frac{Z_1 + Z_2}2$ and $|Y_0|=1$. We now state the following lemma:
\begin{lem}
Under the condition that the following sets are nonempty in $\mathbb R^n, n=2,3$, we have\\
(a) The minimum in $(D_1 \cup D_2)^c$ is the point nearest to $Y_0$.\\
(b) The minimum in $D_1\cap D_2$ is the farest point to $Y_0$.\\
(c) The minimum in $D_1\setminus D_2$ is $N_1$ if $N_1\notin D_2$ and $S_{12}$ if $N_1 \in D_2$.
\end{lem}
\begin{proof}
If $W\in (D_1 \cup D_2)^c$, then 
$$ O(W) = ||W-Z_1||^2 + ||W-Z_2||^2 - d_1^2 - d_2^2 = 2||W-Y_0||^2  + Constant. $$
Thus, we (a)  is proven.

If $W\in D_1 \cap D_2$, then 
$$ O(W) = -||W-Z_1||^2 - ||W-Z_2||^2 + d_1^2 + d_2^2 = -2||W-Y_0||^2  + Constant. $$
Thus, (b) is proven.

If $W\in D_1 \setminus D_2$, then
$$ O(W) = (Z_1-Z_2)\cdot W  +  Constant. $$
Thus, (c) is proven.
\end{proof}

Based on this lemma, we present the following theorem:
\begin{theorem}\label{th:Meas2Psi2}
In the case of two measurements with $\psi(y)=y^2$  and $n=2,3$, the followings minimums occur:\\
(a) If $d_1 \le \frac{\norm{Z_2 - Z_1}}2$, then $X=Y_0$.\\
(b) If $ \frac{\norm{Z_2-Z_1}}2 \le d_1 <\norm{Z_2-Z_1}$ and $d_1 + d_2 \le \norm{Z_2-Z_1}$, then $X=N_1$.\\
(c) If $ d_1 - d_2 < \norm{Z_2-Z_1} < d_1 + d_2$, then $X=S_{12}$.\\
(d) If $d_1 - d_2 \ge \norm{Z_2-Z_1}$, then $X=N_1$.
\end{theorem}

\begin{corollary}
In the two-dimensional case, the number of solutions is either 1 or 2:\\
$$\left\{\begin{array}{ccc}
|X|=1  &\mbox{  if and only if }& \norm{Z_2-Z_1}\ge d_1 + d_2  \mbox{ or }\norm{Z_2-Z_1}\le d_1 - d_2, \\
|X|=2 &\mbox{ if and only if }& d_1 - d_2 <\norm{Z_2-Z_1}< d_1 + d_2 . 
\end{array}\right.$$
\end{corollary}

\begin{corollary}
In the three-dimensional case, the number of solutions is either 1 or $\infty$:\\
$$\left\{\begin{array}{ccc}
|X|=1  &\mbox{  if and only if }& \norm{Z_2-Z_1}\ge d_1 + d_2  \mbox{ or }\norm{Z_2-Z_1}\le d_1 - d_2, \\
|X|=\infty &\mbox{ if and only if }& d_1 - d_2 <\norm{Z_2-Z_1}< d_1 + d_2 . 
\end{array}\right.$$
\end{corollary}

\subsection{The distance error case  ($\psi(y)=y$)}
Note that $Y_0= \overline{Z_1 Z_2} \setminus (D_1 \cup D_2)$ and $|Y_0|=0,1,\infty$. 
\begin{lem}
Under the condition that the following sets are nonempty  in $\mathbb R^n, n=2,3$, we have\\
(a) The minimum in $(D_1 \cup D_2)^c$ is the nearest point to $\overline{Z_1 Z_2}$ as follows :
$$\left\{\begin{array}{ccc}
Y_0      &\mbox{ if }&  d_1 + d_2 \le \norm{Z_1 - Z_2},\\
S_{12} &\mbox{ if }&  d_1 - d_2 < \norm{Z_1 - Z_2} < d_1+d_2,\\
N_1     &\mbox{ if }&  d_1 - d_2  \ge \norm{Z_1-Z_2}.\\
\end{array}\right.$$
(b) The minimum in $D_1\cap D_2$ is the farest point from $\overline{Z_1 Z_2}$ as follows:
$$\left\{\begin{array}{ccc}
S_{12}&\mbox{ if }&  d_1 - d_2 < \norm{Z_1 - Z_2} < d_1+d_2,\\
N_{-2} &\mbox{ if }& d_1 - d_2  \ge \norm{Z_1-Z_2}.\\
\end{array}\right.$$
\\
(c) The minimum in $D_1\setminus D_2$ is  as follows:
$$\left\{\begin{array}{ccc}
N_1  &\mbox{ if }&d_1 + d_2 \le \norm{Z_1 - Z_2},\\
S_{12}&\mbox{ if }&  d_1 - d_2 < \norm{Z_1 - Z_2} < d_1+d_2,\\
\overline{N_{-2} N_1} &\mbox{ if }& d_1 - d_2  \ge \norm{Z_1-Z_2}.\\
\end{array}\right.$$
\end{lem}
\begin{proof}
If $W\in (D_1 \cup D_2)^c$, we have
$$ O(W) = ||W-Z_1|| + ||W-Z_2|| - d_1 - d_2. $$
Thus, (a) is proven.

If $W\in D_1 \cap D_2$, we have
$$ O(W) = -||W-Z_1|| - ||W-Z_2|| + d_1 + d_2 . $$
Thus, (b) is proven.

If $W\in D_1 \setminus D_2$, we have
$$ O(W) =  ||W-Z_2|| - ||W-Z_1|| + d_1 - d_2. $$
Thus, (c) is proven.
\end{proof}

Based on this lemma, we can state the following theorem:
\begin{theorem}\label{th:Meas2Psi1}
For two measurements and distance error minimization with $n=2,3$, the followings minimums occur:\\
(a) If $d_1 + d_2 \le \norm{Z_2 - Z_1}$, then $X=Y_0$.\\
(b) If $ d_1 - d_2 \le \norm{Z_2-Z_1} \le  d_1 + d_2$, then $X=S_{12}$.\\
(c) If $d_1 - d_2 \ge \norm{Z_2-Z_1}$, then $X=\overline{N_{-2} N_1}$.
\end{theorem}

\begin{corollary}
In the two-dimensional case, the number of solutions can be 1,2, or $\infty$:\\
$$\left\{\begin{array}{ccc}
|X|=1  &\mbox{  if and only if }& \norm{Z_2-Z_1}= d_1 + d_2  \mbox{ or }d_1 - d_2, \\
|X|=2  &\mbox{  if and only if }& d_1 - d_2 < \norm{Z_2-Z_1} < d_1 + d_2, \\
|X|=\infty &\mbox{ if and only if }& \norm{Z_2-Z_1}> d_1 + d_2 \mbox{ or } \norm{Z_2-Z_1} <d_1 - d_2. 
\end{array}\right.$$.
\end{corollary}

\begin{corollary}
In the three-dimensional case, the number of solutions can be 1 or $\infty$:\\
$$\left\{\begin{array}{ccc}
|X|=1  &\mbox{  if and only if }& \norm{Z_2-Z_1}= d_1 + d_2  \mbox{ or }d_1 - d_2, \\
|X|=\infty &\mbox{ if and only if }& \norm{Z_2-Z_1} \neq  d_1 + d_2 \mbox{ and } d_1 - d_2. 
\end{array}\right.$$.
\end{corollary}

\subsection{The comparison}
Summarizing the results in Subsection 3.1 and 3.2, the outcomes are compiled in Table \ref{tab:Meas2}. The corresponding figures for each case listed in 
Table 1 are displayed in Fig. \ref{fig:Meas2Dim2} and Fig. \ref{fig:Meas2Dim3}, representing the two- and three-dimensional cases, respectively. 
The corrspondence between these figures and the cases in Table 1 is outlined in Table \ref{tab:Meas2_1}. 
 It can be observed that the numerical results shown in the figures 
align with the theoretical results in the aforementioned subsections.

In Fig.\ref{fig:Meas2Dim2}, which depicts the two dimensional case, the level sets of the objective functions 
are visualized using the `contour' function in Matlab. Here, the minimum value of the objective function is denoted by $m$. 
The blue contours represent level sets ranging from $m$ to $m+3$ with an interval of $0.1$, while 
the cyan countours depict level sets from $m+3$ to $m+50$ with an interval of $1$. The red markers indicate the theoretical minimum points for each case.
For the squared distance case shown in Fig. \ref{fig:Meas2Dim2} (b),(d),(f), and (h), the level sets are vertically confined within two disks that do not intersect with the other disk.
In contrast, for the distance error case, illustrated in Fig. \ref{fig:Meas2Dim2} (a),(c), and (g), the level sets are more horizontally oriented near the solutions. 

The three dimensional results are presented in Fig. {fig:Meas2Dim3} under similar conditions on $Z_1,Z_2, d_1$ and $d_2$. Two spheres are depicted with  black and green dots, while the level sets 
shown on the planes $x=-1, x=1, y=0.1,$ and $z=-0.1$. The gradient directions on grid points, with a step size of $0.5$ in all three axis directions, are illustrated using Matlab's `quiver' command. Streamlines are also plotted in these figures. In the squared distance error cases (Fig. \ref{fig:Meas2Dim3} (b),(d),(f), and (g)), the streamlines inside two balls that are not contained within the other ball run parallel to the $x$-axis, where the centers of the two balls lie. On the other hand, in the distance error cases (Fig. \ref{fig:Meas2Dim3} (a), (c), (e), and (f)), the streamlines inside the black ball, which is not contained within the green ball, converges toward the negative $x$-axis direction.   

\begin{table}
\begin{centering}
\begin{tabular}{|c||c|c|c|c|c|c|}
\hline
                                                                                                                                   &$\psi(y)=y$                        &$\psi(y)=y$      &$\psi(y)=y$         &$\psi(y)=y^2$   &$\psi(y)=y^2$ &$\psi(y)=y^2$\\
  Case                                                                                                                         &Solution $X$                     &$|X|$ : 2D         &$|X|$:3D            &Solution $X$     &$|X|$ : 2D      &$|X|$:3D\\ 
\hline\hline
$d_1< \frac{\norm{Z_1 Z_2}}2$                                                                                &$Y_0$                              &$\infty$            &$\infty$              &$Y_0$               &1                   &1\\
\hline 
$\frac{\norm{Z_1 Z_2}}2\le d_1 <\norm{Z_1 Z_2},$                                                &                                        &                         &                        &                           &                    & \\ 
$d_1+d_2\le\norm{Z_1 Z_2}$                                                                                   &$Y_0$                              &$\infty$ (*1)      &$\infty$ (*1)       &$N_1$               &1                  &1\\
\hline
$d_1 - d_2  < \norm{Z_1 Z_2} < d_1 + d_2$                                                             &$S_{12}$                          &2                      &$\infty$              &$S_{12}$          &2                  &$\infty$\\
\hline 
$d_1 - d_2 \ge \norm{Z_1 Z_2}$                                                                               &$\overline{N_{-2} N_1}$  &$\infty$ (*1)      &$\infty$ (*1)       &$N_1$               &1                 &1\\
\hline                                                                       
\end{tabular}
\caption{The solution location and numbers in 2 measurements. \newline 
*1: The number of $X$ is 1 only for the case $\norm{Z_1 Z_2}=d_1 + d_2$ or $d_1 - d_2$.  }
\label{tab:Meas2}
\end{centering}
\end{table}

\begin{table}
\begin{centering}
\begin{tabular}{|c||c|c|c|c|}
\hline
                                                                                                                                   &$\psi(y)=y$       &$\psi(y)=y$                        &$\psi(y)=y^2$     &$\psi(y)=y^2$\\
  Case                                                                                                                         &$|X|$ : 2D         &$|X|$:3D                           &$|X|$ : 2D           &$|X|$:3D\\ 
\hline\hline
$d_1< \frac{\norm{Z_1 Z_2}}2$                                                                                &Fig.2(a)             &Fig.3(a)                             &Fig.2(b)              &Fig.3(b)\\
\hline 
$\frac{\norm{Z_1 Z_2}}2\le d_1 <\norm{Z_1 Z_2},d_1+d_2\le\norm{Z_1 Z_2}$    &Fig.2(c)             &Fig.3(c)                             &Fig.2(d)              &Fig.3(d)\\ 
\hline 
$d_1 - d_2  < \norm{Z_1 Z_2} < d_1 + d_2$                                                             &Fig.2(e)             &Fig.3(e)                             &Fig.2(f)               &Fig.3(f)\\
\hline 
$d_1 - d_2 \ge \norm{Z_1 Z_2}$                                                                               &Fig.2(g)             &Fig.3(g)                             &Fig.2(h)             &Fig.3(h)\\
\hline                                                                       
\end{tabular}
\caption{The figures in Fig. 2 and Fig. 3 corresponding to each cases in Table 1}
\label{tab:Meas2_1}
\end{centering}
\end{table}

\begin{figure}
\begin{minipage}[t]{8cm}
\centerline{\epsfig{file=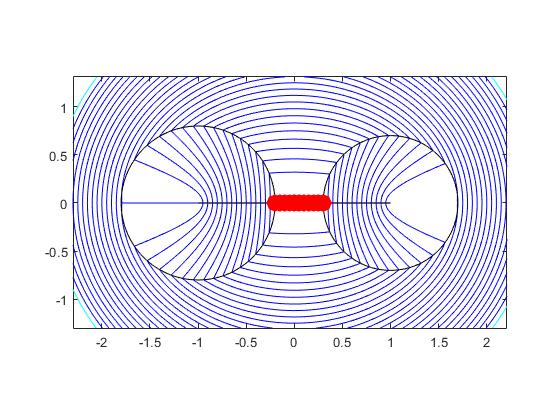, height=4cm,width=5.5cm,clip=1cm}}
\end{minipage}
\begin{minipage}[t]{8cm}
\centerline{\epsfig{file=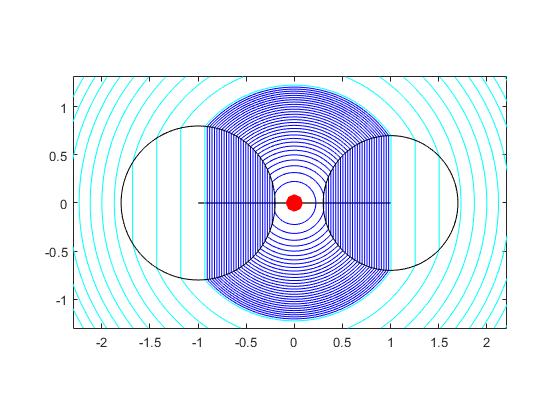, height=4cm,width=5.5cm,clip=1cm}}
\end{minipage}
\begin{center}
(a)\qquad\quad\qquad\qquad\qquad\qquad\qquad\qquad\qquad\qquad\qquad(b)
\end{center}
\begin{minipage}[t]{8cm}
\centerline{\epsfig{file=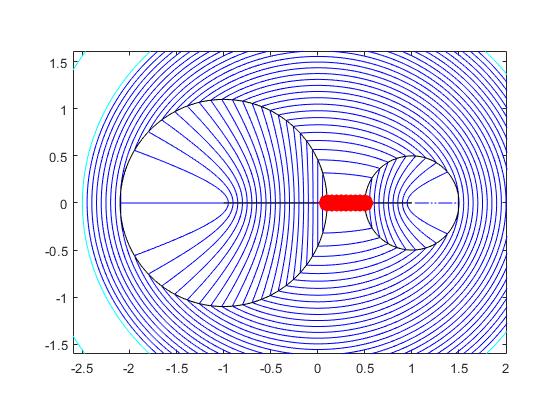, height=4cm,width=5.5cm,clip=1cm}}
\end{minipage}
\begin{minipage}[t]{8cm}
\centerline{\epsfig{file=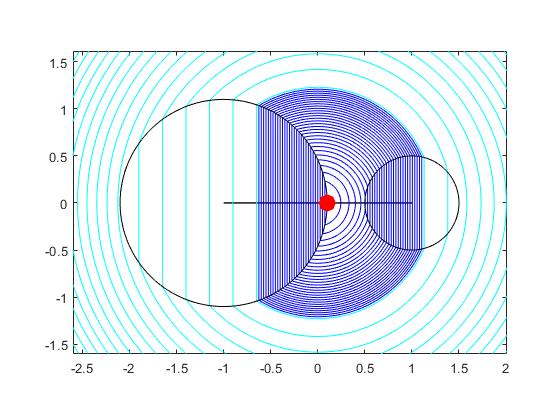, height=4cm,width=5.5cm,clip=1cm}}
\end{minipage}
\begin{center}
(c)\qquad\quad\qquad\qquad\qquad\qquad\qquad\qquad\qquad\qquad\qquad(d)
\end{center}
\begin{minipage}[t]{8cm}
\centerline{\epsfig{file=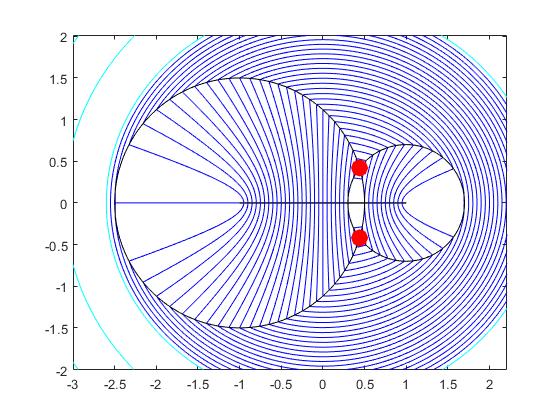, height=4cm,width=5.5cm,clip=1cm}}
\end{minipage}
\begin{minipage}[t]{8cm}
\centerline{\epsfig{file=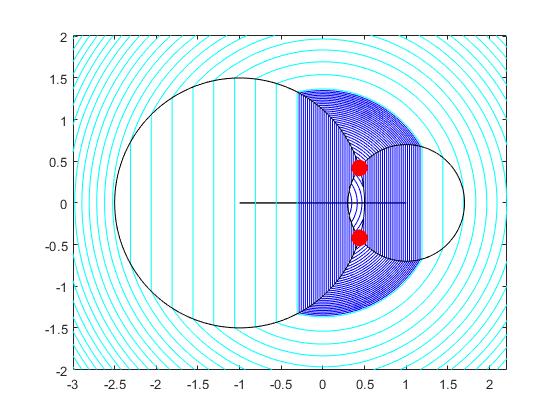, height=4cm,width=5.5cm,clip=1cm}}
\end{minipage}
\begin{center}
(e)\qquad\quad\qquad\qquad\qquad\qquad\qquad\qquad\qquad\qquad\qquad(f)
\end{center}
\begin{minipage}[t]{8cm}
\centerline{\epsfig{file=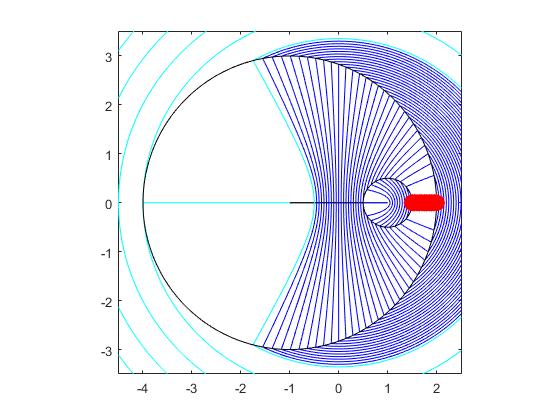, height=4cm,width=5.5cm,clip=1cm}}
\end{minipage}
\begin{minipage}[t]{8cm}
\centerline{\epsfig{file=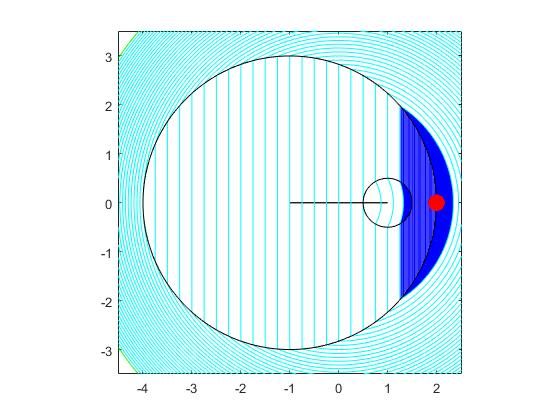, height=4cm,width=5.5cm,clip=1cm}}
\end{minipage}
\begin{center}
(g)\qquad\quad\qquad\qquad\qquad\qquad\qquad\qquad\qquad\qquad\qquad(h)
\end{center}
\caption{The location of the two-diemsional sources for (a,c,e,g) $\psi(y)=y$  and (b,d,f,h) $\psi(y)=y^2$ when 
(a,b) $d_1\le \frac{\norm{Z_1 Z_2}}2$ (c,d) $\frac{\norm{Z_1 Z_2}}2\le d_1 <\norm{Z_1 Z_2}$ and $d_1+d_2\le\norm{Z_1 Z_2}$ 
(e,f) $d_1 - d_2  < \norm{Z_1 Z_2} < d_1 + d_2$ (g,h) $d_1 - d_2 \ge \norm{Z_1 Z_2}$}
\label{fig:Meas2Dim2} 
\end{figure}

\begin{figure}
\begin{minipage}[t]{8cm}
\centerline{\epsfig{file=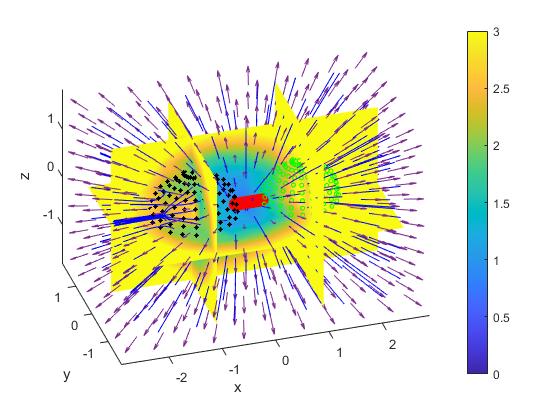, height=3.5cm,width=5cm,clip=1cm}}
\end{minipage}
\begin{minipage}[t]{8cm}
\centerline{\epsfig{file=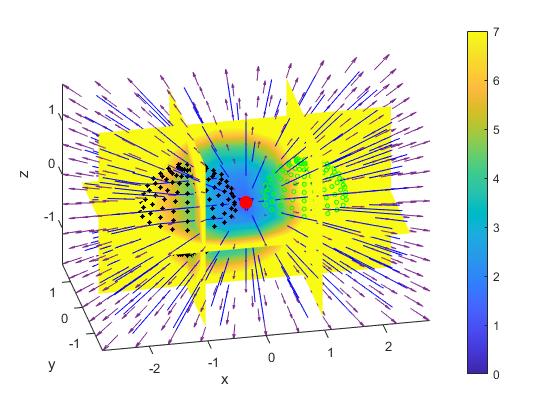, height=3.5cm,width=5cm,clip=1cm}}
\end{minipage}
\begin{center}
(a)\qquad\quad\qquad\qquad\qquad\qquad\qquad\qquad\qquad\qquad\qquad(b)
\end{center}
\begin{minipage}[t]{8cm}
\centerline{\epsfig{file=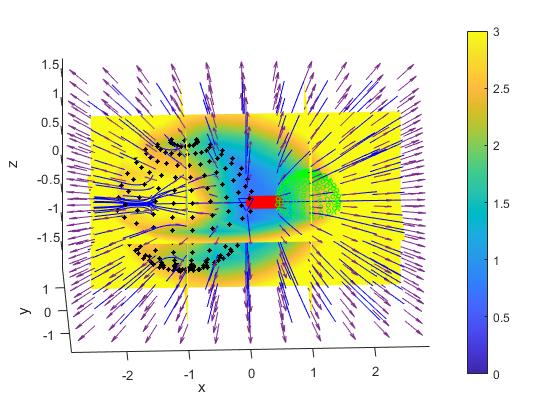, height=3.5cm,width=5cm,clip=1cm}}
\end{minipage}
\begin{minipage}[t]{8cm}
\centerline{\epsfig{file=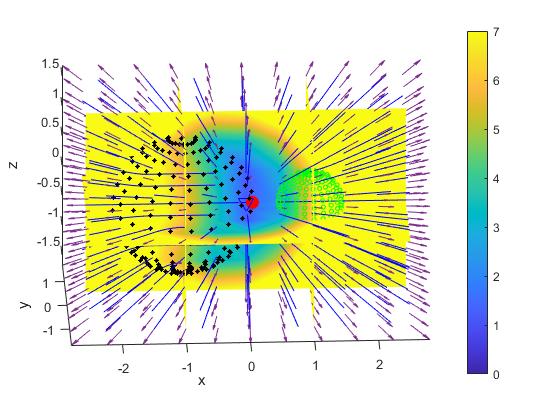, height=3.5cm,width=5cm,clip=1cm}}
\end{minipage}
\begin{center}
(c)\qquad\quad\qquad\qquad\qquad\qquad\qquad\qquad\qquad\qquad\qquad(d)
\end{center}
\begin{minipage}[t]{8cm}
\centerline{\epsfig{file=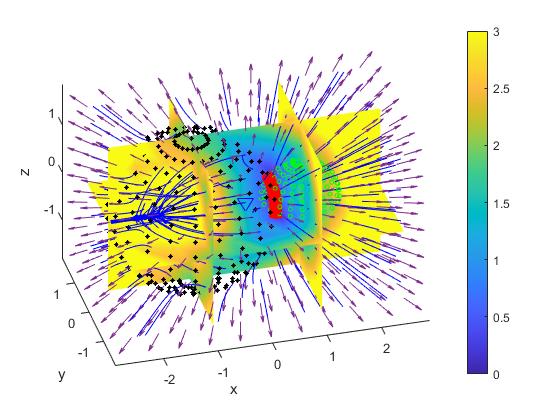, height=3.5cm,width=5cm,clip=1cm}}
\end{minipage}
\begin{minipage}[t]{8cm}
\centerline{\epsfig{file=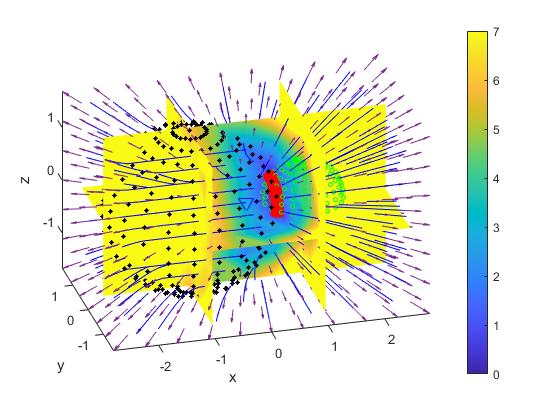, height=3.5cm,width=5cm,clip=1cm}}
\end{minipage}
\begin{center}
(e)\qquad\quad\qquad\qquad\qquad\qquad\qquad\qquad\qquad\qquad\qquad(f)
\end{center}
\begin{minipage}[t]{8cm}
\centerline{\epsfig{file=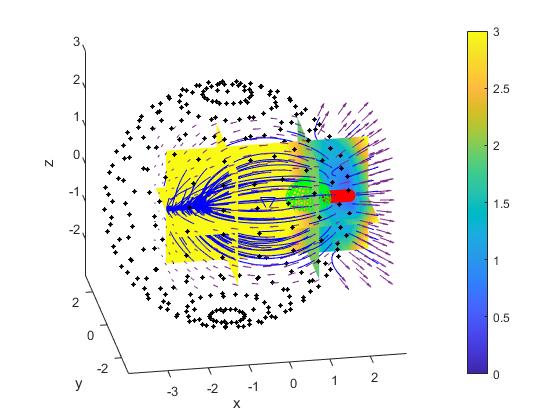, height=3.5cm,width=5cm,clip=1cm}}
\end{minipage}
\begin{minipage}[t]{8cm}
\centerline{\epsfig{file=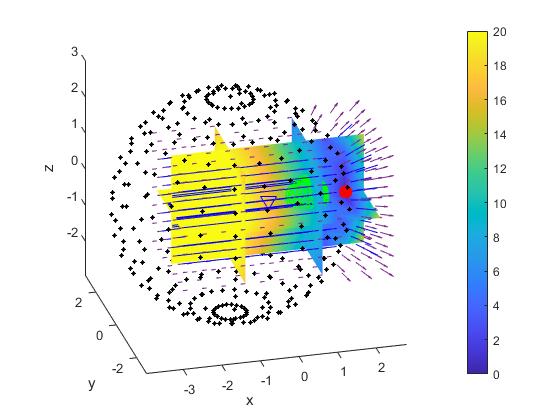, height=3.5cm,width=5cm,clip=1cm}}
\end{minipage}
\begin{center}
(g)\qquad\quad\qquad\qquad\qquad\qquad\qquad\qquad\qquad\qquad\qquad(h)
\end{center}
\caption{
The location of the three-dimensional sources for (a,c,e,g) $\psi(y)=y$  and (b,d,f,h) $\psi(y)=y^2$ when 
(a,b) $d_1\le \frac{\norm{Z_1 Z_2}}2$ (c,d) $\frac{\norm{Z_1 Z_2}}2\le d_1 <\norm{Z_1 Z_2}$ and $d_1+d_2\le\norm{Z_1 Z_2}$ 
(e,f) $d_1 - d_2  < \norm{Z_1 Z_2} < d_1 + d_2$ (g,h) $d_1 - d_2 \ge \norm{Z_1 Z_2}$. }
\label{fig:Meas2Dim3} 
\end{figure}

\section{Three measuremts (I=3)}
Let $1A=A$ and $0A=\phi$ for any set $A$. Define
$$ R_{i j k} = \overline{(iD_1 \cup (1-i)D_1^c) \cap  (jD_2 \cup (1-j)D_2^c) \cap (kD_3 \cup (1-k)D_3^c)}, $$
for $i,j,k \in \{0,1\}$. Also, define $R_l = \cup_{i+j+k=l} R_{ijk}$.

Let $\triangle$ denote the measurement triangle $\triangle Z_1 Z_2 Z_3$, $\alpha$ be the plane containg $\triangle$, and $K=\triangle \cap R_0 \cap \alpha$.

The critical points of the objective function $O$ can be classified as $Y_{ijk}$ where $i,j,k \in \{0,1\}$ and $Y_{ijk}\subset R_{ijk}$. 
Among these, only $Y_0, Y_{100}, Y_{010},$ and $Y_{001}$ can be local minima, while the others may be local maxima. Therefore, 
$$ Y=\{ Y_0\cap R_0, Y_{100}\cap R_{100}, Y_{010}\cap R_{101}, Y_{001}\cap R_{001}\},$$
is the set containing local minima and critical points. Let us define them for the cases $\psi(y)=y^2$ and $\psi(y)=y$.\\

For the squared distance case $\psi(y)=y^2$, the critical points are defined as:
$$\begin{array}{ccc}
Y_0         &=& \frac{Z_1 + Z_2 + Z_3}3,\\
Y_{100}  &=& -Z_1+ Z_2 + Z_3,\\
Y_{010}  &=& Z_1 - Z_2 + Z_3,\\
Y_{001}  &=& Z_1 + Z_2 -Z_3.
\end{array}$$
If $Y_0 \in R_0, Y_{100}\in R_{100}, Y_{010}\in R_{010},$ or $Y_{001}\in R_{001}$, then these points are critical points for the objective function $O$. 
It is proven in $\cite{KL2}$ that $|\{ Y_0\cap R_0, Y_{100}\cap R_{100}, Y_{010}\cap R_{010}, Y_{001}\cap R_{001} \}|\le 1$, meanig that the set contains either 0 or 1 point. 

For the distance error case $\psi(y)=y$, $Y_0$ is defiend as the Fermat point of the triangle $\triangle Z_1 Z_2 Z_3$. 
To describe $Y_{100}$, consider the plane $\alpha$ containing $Z_1, Z_2,$ and $Z_3$.  
In this plane, two half planes are divided by the line $\overleftrightarrow{Z_2 Z_3}$.  
Consider a point $\tilde Z_1$ in the half plane containg $Z_1$ such that $\tilde Z_1 Z_2 Z_3$ is a equilateral triangle. 
Let $C$ be the circumscribed circle of $\triangle \tilde Z_1 Z_2 Z_3$ and $\tilde C$ be the arc of $C$ that lies in the other half-plane 
and connects $Z_2$ and $Z_3$. $Y_{100}$ is defined as the intersection of the line $\tilde Z_1 Z_1$ and $\tilde C$, which can either be empty or consist of a single one point
if $\tilde Z_1 \neq Z_1$, or be the entire arc if $\tilde Z_1 = Z_1$. Similarly, $Y_{010}$ and $Y_{001}$ are defined.

If $\triangle Z_1 Z_2 Z_3$ is not eqilateral, $|Y|$ equals 0 or 1, and if $\triangle Z_1 Z_2 Z_3$ is equilateral, $|Y|=\infty$. \\

Let 
$$ S_{123}= S_1 \cap S_2 \cap S_3, $$
and
$$ S_{ij} = S_i \cap S_j,$$
for $i,j \in \{1,2,3\}$.

For both cases of $\psi$, $Y_0$ lies in $\overline \triangle$. Define $S_{12}^+$ and $S_{12}^-$ as the points in $S_{12}$ which are closest to and farest from $Y_0$, respectively. 
Note that $S_{12}^+$ and $S_{12}^-$ lie in $alpha$. Similarily, $S_{23}^\pm$ and $S_{31}^\pm$ are defined.

The definitions of $N_i,$ for $i=1,2,3$, as members of $S_i$, vary depending on $\psi$ and will be explained in the following subsections.

The results in this section are derived from the  papers \cite{KL2,KL2Cor} for $\psi(y)=y^2$ and the paper \cite{KL1} for $\psi(y)=y$ in two-dimensional case. 
The following two lemmas illustrate the difference between the two- and three- dimensional cases:
\begin{lem}\label{le:2D}
In the two dimensional case with $n=2$, the following hold:
\begin{itemize}
\item{If $Z_1,Z_2,$ and $Z_3$ are not collinear, $|S_{123}|=0,$ or $1$.
}
\item{If $|S_{123}|=1$, then $X=S_{123}$.
}
\item{If $R_{100}$ is nonempty, then it is either connected.or has two connected components.
}
\item{If all of $R_{100},R_{010},$ and $R_{001}$ are nonempty, then $|S_{123}|=0$.
}
\end{itemize}
\end{lem}
The proof is straightforward and can be found in \cite{KL1, KL2,KL2Cor}. However, in the three-dimensional case, the situation differs significantly from the two-dimensional case,
as described below:

\begin{lem}\label{le:3D}
In the three dimensional case with $n=3$, the following hold:
\begin{itemize}
\item{If $Z_1,Z_2,Z_3$ does not lies on the same line, $|S_{123}|=0,1,$ or $2$.}
\item{If $|S_{123}|=1$ or $2$, then $X=S_{123}$.}
\item{If $R_{100}$ is nonempty, then it is connected.}
\item{If all of $R_{100},R_{010},R_{001}$ are nonempty, then $|S_{123}|=2$.}
\end{itemize}
\end{lem}
\begin{proof}
$S_{12}$ can be empty or one point or a circle. $S_{123}$ could be empty, a single point, or two points. If $S_{123}\neq \phi$, then 
$O(S_{123}) = 0$ and $X=S_{123}$.

Assume $R_{100}$ is nonempty. Then, one of two points $Z_1 - d_1 \frac{Z_2-Z_1}{\norm{Z_2-Z_1}}$ or $Z_1 - d_1 \frac{Z_3-Z_1}{\norm{Z_3-Z_1}}$, 
which are the farest points from $Z_2$ and $Z_3$ respectively, is contained in $R_{100}$. Therefore, $R_{100}\cap\alpha$ is also nonempty.

If $D_2 \cap D_3 =\phi, D_2 \subset D_3,$ or $D_3\subset D_2$, then $R_{100}$ is easily proven to be connected. 
Otherwise, we have $|d_2-d_3| \le \norm{Z_2 Z_3} \le d_2 + d_3$. 
Let the plane passing through $Z_1,Z_2,$ and $Z_3$ be denoted as $\alpha$.
Since the region $R_{100}$ is symmetric with respect to the plane $\alpha$, we only need to consider the half-space divided by $\alpha$.
Let $\beta$ be the plane containing $S_{23}$ and orthogonal to $\alpha$.
Let $B$ be the  foot of the perpendicular from $Z_1$ to the plane $\beta$. Then, $D_1 \cap \beta$ is the disk centered at $B$ with radius $\sqrt{d_1^2 - \norm{Z_1 B}^2}$
and this disk contains $S_{23}$. Hence, in this case also, $R_{100}$ is also connected. Thus, in all the cases, if $R_{100}$ is nonempty, it is connected.

If all of $R_{100},R_{010},R_{001}$ are nonempty, then $R_{100}\cap\alpha,R_{010}\cap\alpha,R_{001}\cap\alpha$ are all nonempty, as shown previously. 
$R_3 \cap \alpha$ forms a spherical triangle in $\alpha$. 
Let $\alpha^k$ be the hyperplane parallel to $\alpha$ at a distance $k$, and let $Z_1^k, Z_2^k,Z_3^k$ be the intersection points between $\alpha^k$ and the perpendicular to $\alpha$ passing through $Z_1,Z_2,Z_3$, respectively. The radius of a circle in the plane $\alpha^k$ centered at $Z_j^k$ (for $j=1,2,3$) becomes $\sqrt{ max(d_j^2 - k^2,0)}$
Therefore, there exists a number $k>0$ such that $R_3\cap \alpha_k$ is a singel point. By symmetry with respect to the plane $\alpha$, there are two points such the line joining them is
orthogonal to $\alpha$. These two points must be $S_{123}$. 
\end{proof}


\subsection{The squared distance error case  ($\psi(y)=y^2$)}

By Lemma \ref{le:2D} and Lemma \ref{le:3D}, if $S_{123}\neq\phi$, then $X=S_{123}$.
Form now on, consider the nontrivial case where $S_{123}=\phi$. 
The results vary depending on whether the dimension is $n=2$ or $n=3$. 

In two dimension, the following cases arise:
\begin{enumerate}
\item{$R_3 = \phi$}
\item{$R_3 \neq \phi$  and  one of $R_{100}, R_{010}, R_{001}$ is empty.}
\item{$R_3 \neq \phi$  and  one of $R_{100}, R_{010}, R_{001}$ is nonempty but disconnected.}
\item{$R_3 \neq \phi$  and  all of $R_{100}, R_{010}, R_{001}$ are nonempty and connected.}
\end{enumerate}
In three dimensions, with the aid of Lemma \ref{le:3D}, we enounter the following simple cases:
\begin{enumerate}
\item{$R_3 = \phi $}
\item{$R_3 \neq \phi$ and  one of  $R_{100}, R_{010}, R_{001}$ is empty.}
\end{enumerate}

In the two-dimensional case 3, $X=X_{100}$ and $|X|=1$ \cite{KL2}. 
In case 4, $X\subset \{ S_{12}^\pm, S_{23}^\pm, S_{31}^\pm\}$ and $|X|\le 5$ \cite{KL2Cor}. 
In this subsection, focus is placed on cases 1 and 2 in both two or three dimensions, which covers all the possible scenarios in three dimensions.
In case 2, assume without loss of generality that $R_{001}=\phi$, which means $D_3\subset D_1\cup D_2$.
Therefore, the analysis in this subsection is restricted to the cases where $R_3=\phi$ or $D_3\subset D_1\cup D_2$. 
Define
$$ N_{ij} = Z_i + d_i  \frac{Y_j - Z_i}{\norm{Y_j-Z_i}}, \quad    N_i := N_{ii}, \quad  i,j \in \{1,2,3\}. $$

Following \cite{KL2}, the following results hold in both two and three diemensions:\\

\begin{itemize}
\item{If $D_3 \subset D_1 \cap D_2$ and $S_{12}\neq\phi$, then $X=S_{12}^+$.}
\item{If $D_3\subset D_1$ and $D_2\subset D_1$, then  $X=X_{100}$, meaning\\
$$\left\{\begin{array}{ccl}
X = Y_{100}  &\mbox{ if }& Y_{100}\in R_{100},\\
X = N_1        &\mbox{ if }& Y_{100}\notin D_1,\\
X=N_{21}     &\mbox{ if }& Y_{100}\in D_2, N_{21}\in R_{100},\\
X=S_{23}^+ &\mbox{ if }& Y_{100}\in D_2, N_{21}\in D_2.
\end{array}\right.$$}
\item{If $K$ is nonempty, then $R_3=\phi$ and $X=X_0$, meaning,\\
$$\left\{\begin{array}{ccl}
X=Y_0         &\mbox{ if }& Y_0\in R_0,\\
X=N_1         &\mbox{ if }& Y_0\in D_1, N_1\in R_0,\\
X=S_{12}^+ &\mbox{ if }& Y_0\in D_1, N_1\in D_2.\\
\end{array}\right.$$}
\item{If $K=\phi$ and $R_3=\phi$, then there exists a ball(3D) or a disk(2D) that intersects the other two balls or two disks. 
Let this ball or disk be $D_1$. Then, we have $X=X_{100}$, meaning:\\
$$\left\{\begin{array}{ccl}
X=Y_{100}    &\mbox{ if }& Y_{100}\in R_{100},\\
X=N_1          &\mbox{ if }& Y_{100}, N_1 \in R_0,\\
X=N_{21}     &\mbox{ if }& Y_{100}\in D_2, N_{21}\in R_{100},\\
X=S_{12}^+ &\mbox{ if }& Y_{100}\in R_0, N_1\in D_2  \mbox{ or } Y_{100}\in D_2, N_{21}\in R_0.\\
\end{array}\right.$$}
\end{itemize}

\begin{corollary}
In two dimensional case with $n=2$, the number of solutions can be 1, 2, 3, 4, or 5:\\
$$\left\{\begin{array}{rl}
|X|\le 5, &X\subset \{S_{12}^\pm, S_{23}^\pm, S_{31}^\pm \} \\
&\mbox{ if } |S_{123}|=0, R_3\neq\phi, \mbox{ and all of } R_{100}, R_{010}, R_{001}  \mbox{ are nonempty and connected,}\\
|X|=1 &\mbox{ Otherwise. } 
\end{array}\right.$$
\end{corollary}

\begin{corollary}\label{co:12}
In three dimensional case with $n=3$, the number of solutions can be 1 or 2:\\
$$\left\{\begin{array}{ccc}
|X|=1 &\mbox{ if }& |S_{123}|=0,1, \\
|X|=2 &\mbox{ if }& |S_{123}|=2.
\end{array}\right.$$
\end{corollary}

Suppose that $\norm{Z_1 Z_3}=\norm{Z_2 Z_3}, d_1=d_2,$ and $d_3^2 = d_1^2 - ||Z_1 Z_3||^2$, and define 
$$ P= \frac{||Z_1 Z_2||^2}4 + \left(||Z_1 Z_3||^2 - \frac{||Z_1 Z_2||^2}4\right)
                             \left( \frac{ ||Z_1 Z_3||^2 + ||Z_1 Z_2||^2}{ ||Z_1 Z_3||^2 - ||Z_1 Z_2||^2}\right)^2. $$
We also have the following detailed interesting non-uniqueness results in the two dimensional case \cite{KL2Cor}:
$$\left\{\begin{array}{ll}
|X|=5, X=\{S_{12}^+,S_{23}^\pm,S_{31}^\pm\}  &\mbox{ if } ||Z_1 Z_3||> ||Z_1 Z_2|| \mbox{ and } d_1^2 = P.\\
|X|=4, X=\{S_{23}^\pm,S_{31}^\pm\}                  &\mbox{ if } ||Z_1 Z_3||> ||Z_1 Z_2|| \mbox{ and } d_1^2  > P.\\
|X|=1, X=S_{12}^+                                                 &\mbox{ if } ||Z_1 Z_3||\le ||Z_1 Z_2|| \mbox{ or } d_1^2  \le  P.           
\end{array}\right.$$

Futhermore, the first condition is both sufficient and necessary for $|X|=5$, with reordering if neededy. Specifically,
$$ | X|=5 \mbox{ if and only if } ||Z_1 Z_3||=||Z_2 Z_3||> ||Z_1 Z_2||,  d_3^2 = d_1^2 - ||Z_1 Z_3||^2,  \mbox{ and }d_1^2 = P, $$
where reordering may be requried. In this scenario, $X$ is given by $X=X=\{S_{12}^+,S_{23}^\pm,S_{31}^\pm\}$, again with reordering if necessary.

The exact location of the soultion $X$ is not precisely determined in the two dimensional case when 
$$
|S_{123}|=0 \mbox{ and } R_3\neq\phi \mbox{ and all of } R_{100}, R_{010}, R_{001}  \mbox{ are nonempty and connected}.
$$
In this case, it is only known that $X\subset \{S_{12}^\pm, S_{23}^\pm, S_{31}^\pm \}$.
However, in the isoceles measurement cases, where $||Z_1 Z_3||=||Z_2 Z_3||$ and $d_1 = d_2$, the exact locaion and the number of the solution are known \cite{KL2Cor}.

\subsection{The distance error case ($\psi(y)=y$)}
By Lemma \ref{le:2D} and Lemma \ref{le:3D}, if $S_{123}\neq\phi$, then $X=S_{123}$.
Now, consider only the case when $S_{123}=\phi$. 
Define $K=R_0\cap \triangle$ and 
$$T= \alpha\cap\left( S_1 \cap \overline{Z_2 Z_3}\cap \partial R_{100} \right) \cup \left( S_2 \cap \overline{Z_1 Z_3}\cap \partial R_{010} \right)  \cup \left( S_3 \cap \overline{Z_1 Z_2}\cap \partial R_{001} \right).$$
Note that $T\neq\phi$, then $K\neq\phi$.
Next, let $N_1$ be defined as the point on $S_1$ where the sum of distances to $Z_2$ and $Z_3$ is minimized. The point is uniquely defined if $S_1 \cap \overline{Z_2 Z_3}=\phi$,
but there may be two such points, otherwise.   

The number and location of solutions, whether in two or three dimensions, can be summarized as follows (the proof for the two dimensional case is provided in \cite{KL1}):
$$\left\{\begin{array}{rcl}
X = S_{123}                                                                &\mbox{ if }&  S_{123}\neq \phi, \\
X = Y_0                                                                       &\mbox{ if }&  Y_0 \in R_0, \\
X = Y_{100}                                                                &\mbox{ if }&  Y_{100}\in R_{100},\\
X = T \mbox{ and }  |X|=1 \mbox{ or } 2                      &\mbox{ if }&  T\neq\phi,\\
X = N_1                                                                      &\mbox{ if }&  T=\phi, K\neq\phi,  Y_0\in R_{100}, N_1\in R_0\\
X = S_{12}^+                                                              &\mbox{ if }&  \left(T=\phi,  K\neq\phi, Y_0\in R_{110}\right) \mbox{ or } \left( T=\phi, K\neq\phi,  Y_0\in R_{100}, N_1\in D_2\right)\\ 
X = \tilde C\cap R_{100} \mbox{ and }|X|=\infty         &\mbox{ if }&  \triangle \mbox{ is equilateral and } \frac{\sqrt 3}2 d_1 >\norm{Z_1-Z_2} > d_2 + d_3. 
\end{array}\right.$$

These cases do not cover all possible scenarios, but they do address all situations where $K\neq\phi$. For the case where $K=\phi$, much remains unknown.

\subsection{The comparison}

In the two dimensional case, we compared the solutions for the distance error ($\psi(y)=y$)  and squred distance error ($\psi(y)=y^2$) cases, as illustrated in Figures \ref{fig:Meas3Dim2_1} and \ref{fig:Meas3Dim2_2}.
The detailed settings for the measurement triangle, the measurement data, and corresponding solutions in the figures are explained in Table \ref{tab:Meas3Dim2}.

\begin{table}
\begin{centering}
\begin{tabular}{|c||c|c|}
\hline
                                                                                  &$\psi(y) =y$ \cite{KL1}                                                                            &$\psi(y)=y^2$ \cite{KL2}\\
\hline \hline
$S_{123}\neq\phi$                                                   &$X=S_{123}$ (Fig.\ref{fig:Meas3Dim2_1}(a))                                          &$X=S_{123}$(Fig.\ref{fig:Meas3Dim2_1}(b))\\
\hline
$Y_0 \in R_0$                                                           &$X=Y_0$ (Fig.\ref{fig:Meas3Dim2_1}(c))                                                &$X=Y_0$(Fig.\ref{fig:Meas3Dim2_1}(d))\\
\hline
$Y_0\in R_{010}, T=\phi, K\neq\phi, N_2\in R_0$ &$X=N_2$(Fig.\ref{fig:Meas3Dim2_1}(e))                                                 &$X=N_2$(Fig.\ref{fig:Meas3Dim2_1}(f))\\
\hline
$Y_{100}\cap R_{100}\neq\phi$ and $\triangle$ is equilateral
                                                                                  &$X=Y_{100}\cap R_{100}, |X|=\infty$(Fig.\ref{fig:Meas3Dim2_1}(g))    &$X=Y_{010}, |X|=1$(Fig.\ref{fig:Meas3Dim2_1}(h))\\
\hline                                                                                   
$|T|=2$                                                                      &$X=T$(Fig.\ref{fig:Meas3Dim2_2}(a))                                                     &$X=N_1$(Fig.\ref{fig:Meas3Dim2_2}(b))\\   
\hline
$|T|=1, N_1\in D_2$                                                &$X=T\cap R_0$(Fig.\ref{fig:Meas3Dim2_2}(c))                                     &$X=S_{12}^+$(Fig.\ref{fig:Meas3Dim2_2}(d))\\
\hline                                             
\end{tabular}
\caption{Comparison for the solutions for $\psi(y)$ and $\psi(y)=y^2$ in two dimensions. }
\label{tab:Meas3Dim2}
\end{centering}
\end{table}

It is important to note that the definitons of $Y_0, N_2,$ and $ Y_{100}$ vary depending on the choice of $\psi$. 
For instance, $Y_0$ is referred as the `Fermat point' for $\psi(y)=y$ and  as the `center of gravity' for $\psi(y)=y^2$. 
Additionally, in the case of an equilateral measurement triangle, $|Y_{100}|$ is $\infty$ for $\psi(y)=y$ and $1$ for $\psi(y)=y^2$.

In Figures \ref{fig:Meas3Dim2_1}, \ref{fig:Meas3Dim2_2}, and \ref{fig:Meas3Dim2_Psi2}, which depicts the two dimensional case and three measurements, the level sets of the objective functions are visualized using the `contour' function in Matlab. Here, the minimum value of the objective function is denoted by $m$. 
\begin{itemize}
\item{ The blue contours represent level sets ranging from $m$ to $m+3$ with an interval of $0.5$, } 
\item{ The cyan countours represent level sets from $m+3$ to $m+50$ with an interval of $2$. } 
\item{ The green countours represent level sets from $m+50$ to $m+200$ with an interval of $10$.} 
\end{itemize}
The red markers indicate the theoretical minimum points for each case.

The level sets for the distance case shown in Fig. \ref{fig:Meas3Dim2_1} (a),(c),(e),(g), Fig. \ref{fig:Meas3Dim2_2} (a) and (c), 
are more dependent on the vertices $Z_1,Z_2, $ and $Z_3$ compared to the level sets for the squared distance error case, shown in Fig. \ref{fig:Meas3Dim2_1} (b),(d),(f),(h), Fig. \ref{fig:Meas3Dim2_2} (b) and (d),

The location of sources when $K=\phi$ is only known for $\psi(y)=y^2$. 
In Fig. \ref{fig:Meas3Dim2_Psi2}, we illustrate the case when $|X|=2,3,4,5$ for flat, sharp isoscles and equilateral triangles. For detailed explanations of these cases, refer to \cite{KL2Cor} 
As shown in Corollary \ref{co:12}, it is important to note that the number of solutions for $\psi(y)=y^2$ and three dimensional case, is less than 2.

\begin{figure}
\begin{minipage}[t]{8cm}
\centerline{\epsfig{file=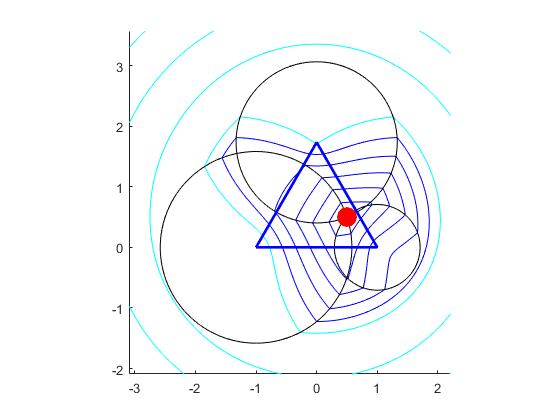, height=4cm,width=5.5cm,clip=1cm}}
\end{minipage}
\begin{minipage}[t]{8cm}
\centerline{\epsfig{file=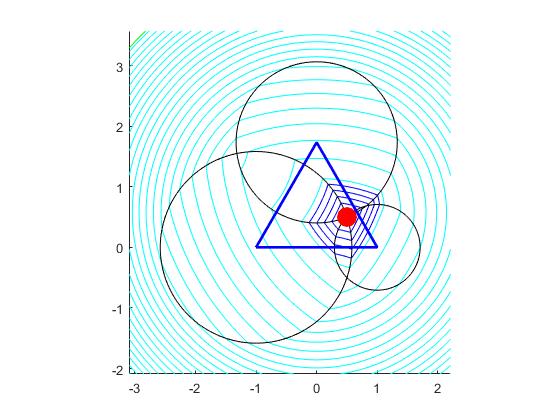, height=4cm,width=5.5cm,clip=1cm}}
\end{minipage}
\begin{center}
(a)\qquad\quad\qquad\qquad\qquad\qquad\qquad\qquad\qquad\qquad\qquad(b)
\end{center}
\begin{minipage}[t]{8cm}
\centerline{\epsfig{file=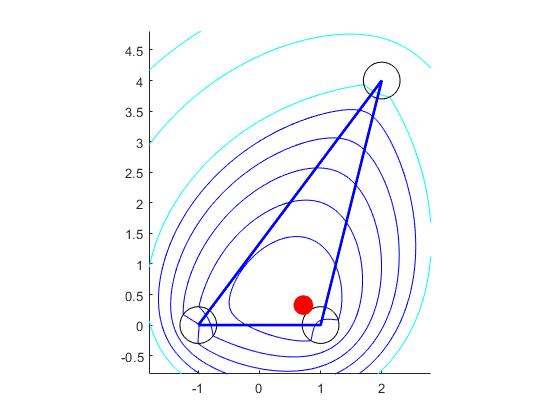, height=4cm,width=5.5cm,clip=1cm}}
\end{minipage}
\begin{minipage}[t]{8cm}
\centerline{\epsfig{file=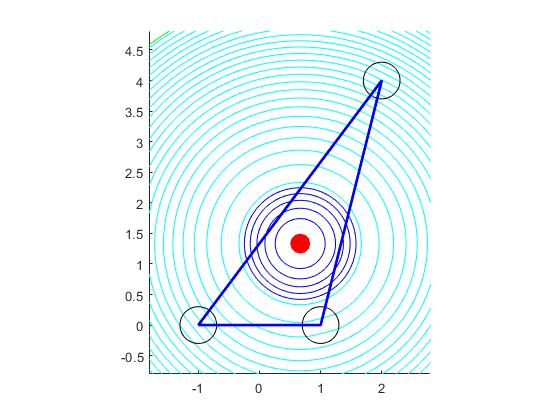, height=4cm,width=5.5cm,clip=1cm}}
\end{minipage}
\begin{center}
(c)\qquad\quad\qquad\qquad\qquad\qquad\qquad\qquad\qquad\qquad\qquad(d)
\end{center}
\begin{minipage}[t]{8cm}
\centerline{\epsfig{file=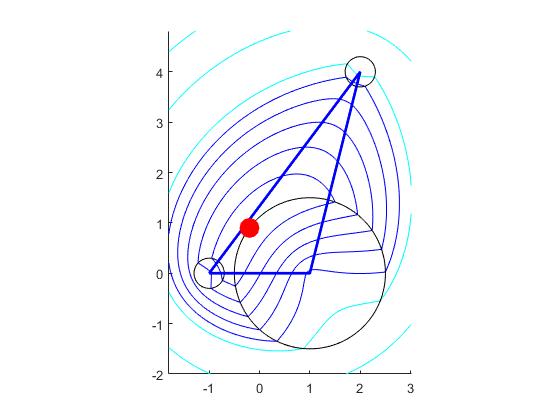, height=4cm,width=5.5cm,clip=1cm}}
\end{minipage}
\begin{minipage}[t]{8cm}
\centerline{\epsfig{file=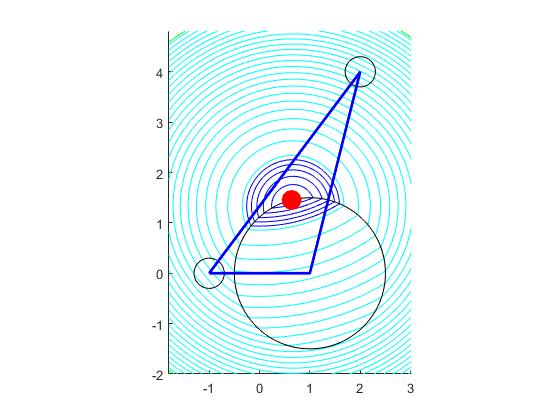, height=4cm,width=5.5cm,clip=1cm}}
\end{minipage}
\begin{center}
(e)\qquad\quad\qquad\qquad\qquad\qquad\qquad\qquad\qquad\qquad\qquad(f)
\end{center}
\begin{minipage}[t]{8cm}
\centerline{\epsfig{file=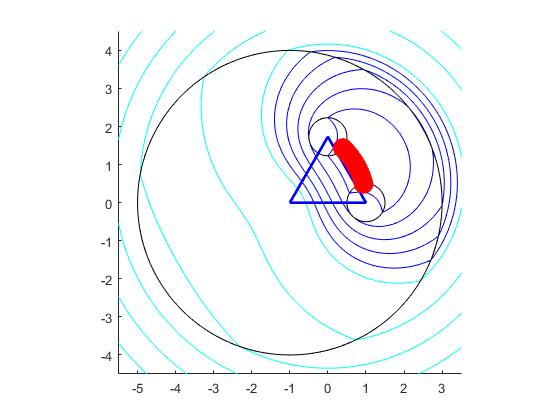, height=4cm,width=5.5cm,clip=1cm}}
\end{minipage}
\begin{minipage}[t]{8cm}
\centerline{\epsfig{file=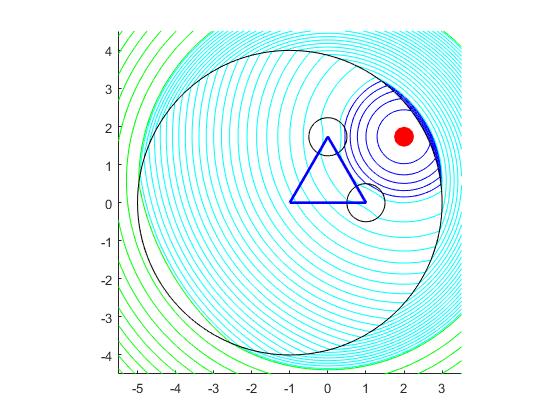, height=4cm,width=5.5cm,clip=1cm}}
\end{minipage}
\begin{center}
(g)\qquad\quad\qquad\qquad\qquad\qquad\qquad\qquad\qquad\qquad\qquad(h)
\end{center}
\caption{The source locations for cases (a,c,e,g) are based on $\psi(y)=y$, while those for cases (b,d,f,h) are based on $\psi(y)=y^2$, 
The conditions for each case are (a,b) $d_1\le \frac{\norm{Z_1 Z_2}}2$ (c,d) $\frac{\norm{Z_1 Z_2}}2\le d_1 <\norm{Z_1 Z_2},d_1+d_2\le\norm{Z_1 Z_2}$ 
(e,f) $d_1 - d_2  < \norm{Z_1 Z_2} < d_1 + d_2$ (g,h) $d_1 - d_2 \ge \norm{Z_1 Z_2}$. For detailed description of these cases see Tab. \ref{tab:Meas3Dim2}. 
}
\label{fig:Meas3Dim2_1} 
\end{figure}

\begin{figure}
\begin{minipage}[t]{8cm}
\centerline{\epsfig{file=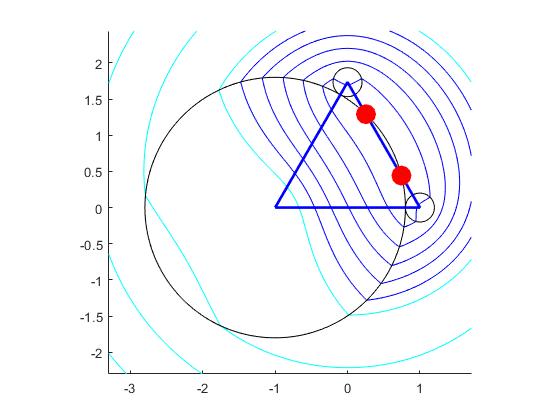, height=4cm,width=5.5cm,clip=1cm}}
\end{minipage}
\begin{minipage}[t]{8cm}
\centerline{\epsfig{file=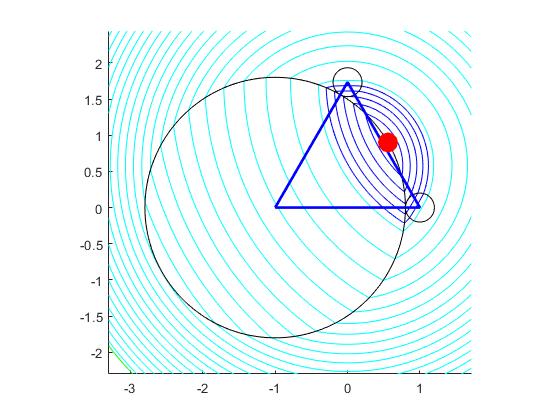, height=4cm,width=5.5cm,clip=1cm}}
\end{minipage}
\begin{center}
(a)\qquad\quad\qquad\qquad\qquad\qquad\qquad\qquad\qquad\qquad\qquad(b)
\end{center}
\begin{minipage}[t]{8cm}
\centerline{\epsfig{file=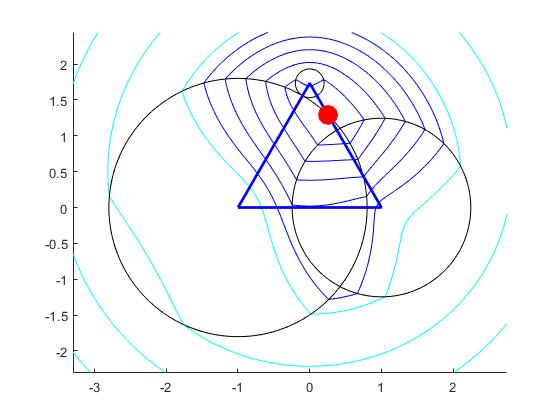, height=4cm,width=5.5cm,clip=1cm}}
\end{minipage}
\begin{minipage}[t]{8cm}
\centerline{\epsfig{file=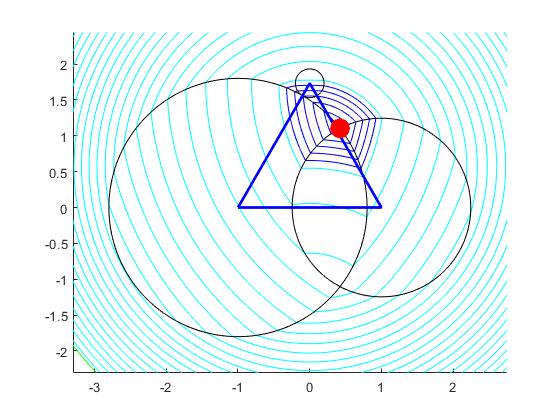, height=4cm,width=5.5cm,clip=1cm}}
\end{minipage}
\begin{center}
(c)\qquad\quad\qquad\qquad\qquad\qquad\qquad\qquad\qquad\qquad\qquad(d)
\end{center}
\caption{The source locations for cases (a,c,e,g) are based on $\psi(y)=y$, while those for cases (b,d,f,h) are based on $\psi(y)=y^2$, 
The conditions for each case are (a,b) $|T|=2$ and  (c,d) $|T|=1, N_1\in D_2$. For detailed description of these cases see Tab. \ref{tab:Meas3Dim2}. 
}
\label{fig:Meas3Dim2_2} 
\end{figure}

The number of solutions in two-dimensional cases with squared distance error can be $2,3,4,$ or $5$ . As proved in \cite{KL2Cor}, the maximum number of solutions is 5.
Examples illustrating these cases are provided in Fig. \ref{fig:Meas3Dim2_Psi2}.

\begin{figure}
\begin{minipage}[t]{8cm}
\centerline{\epsfig{file=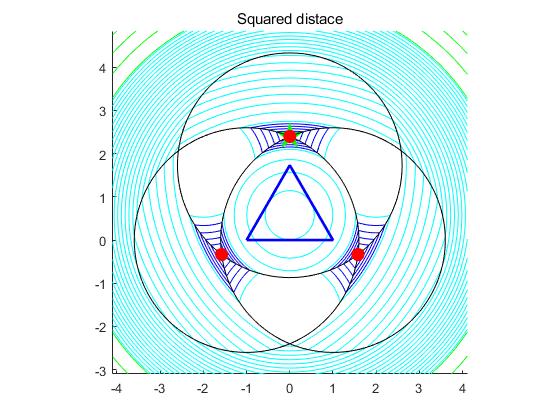, height=4cm,width=5.5cm,clip=1cm}}
\end{minipage}
\begin{minipage}[t]{8cm}
\centerline{\epsfig{file=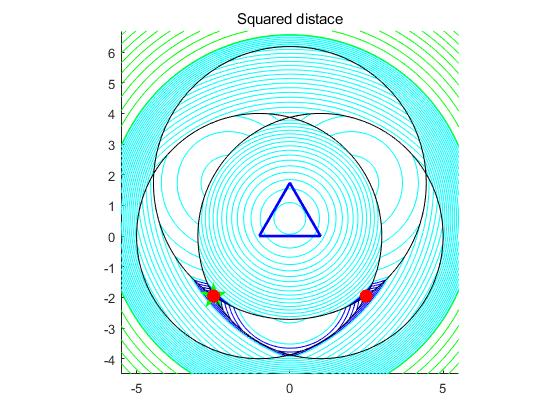, height=4cm,width=5.5cm,clip=1cm}}
\end{minipage}
\begin{center}
(a)\qquad\quad\qquad\qquad\qquad\qquad\qquad\qquad\qquad\qquad\qquad(b)
\end{center}
\begin{minipage}[t]{8cm}
\centerline{\epsfig{file=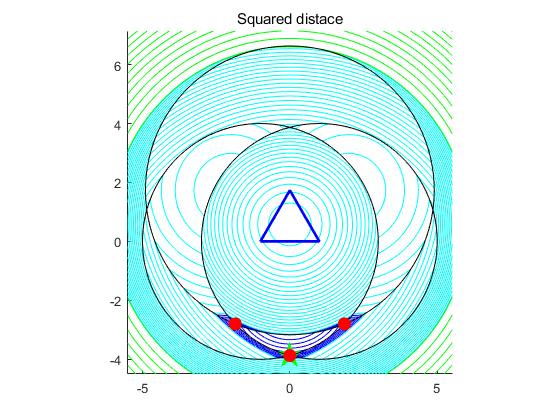, height=4cm,width=5.5cm,clip=1cm}}
\end{minipage}
\begin{minipage}[t]{8cm}
\centerline{\epsfig{file=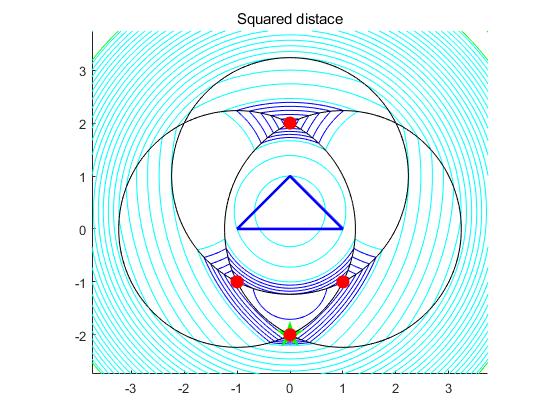, height=4cm,width=5.5cm,clip=1cm}}
\end{minipage}
\begin{center}
(c)\qquad\quad\qquad\qquad\qquad\qquad\qquad\qquad\qquad\qquad\qquad(d)
\end{center}
\begin{minipage}[t]{8cm}
\centerline{\epsfig{file=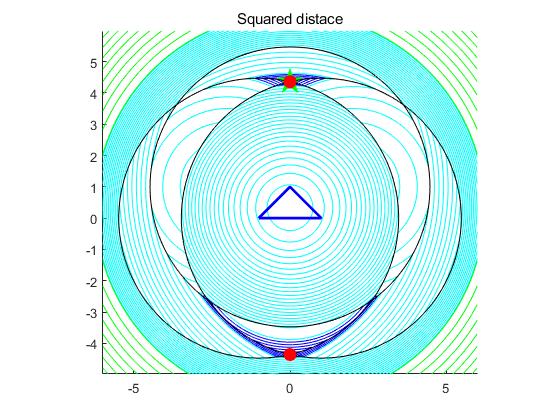, height=4cm,width=5.5cm,clip=1cm}}
\end{minipage}
\begin{minipage}[t]{8cm}
\centerline{\epsfig{file=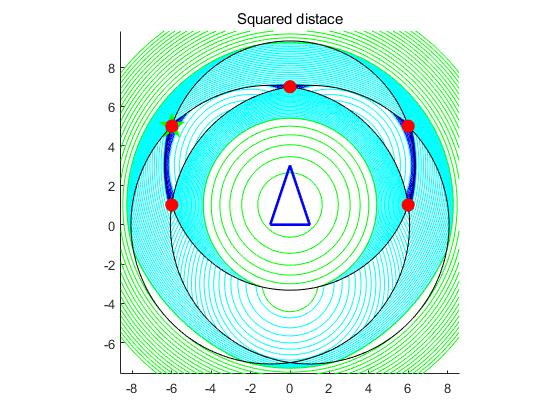, height=4cm,width=5.5cm,clip=1cm}}
\end{minipage}
\begin{center}
(e)\qquad\quad\qquad\qquad\qquad\qquad\qquad\qquad\qquad\qquad\qquad(f)
\end{center}
\begin{minipage}[t]{8cm}
\centerline{\epsfig{file=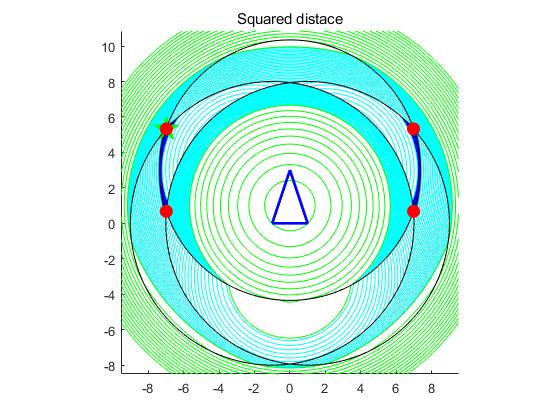, height=4cm,width=5.5cm,clip=1cm}}
\end{minipage}
\begin{minipage}[t]{8cm}
\centerline{\epsfig{file=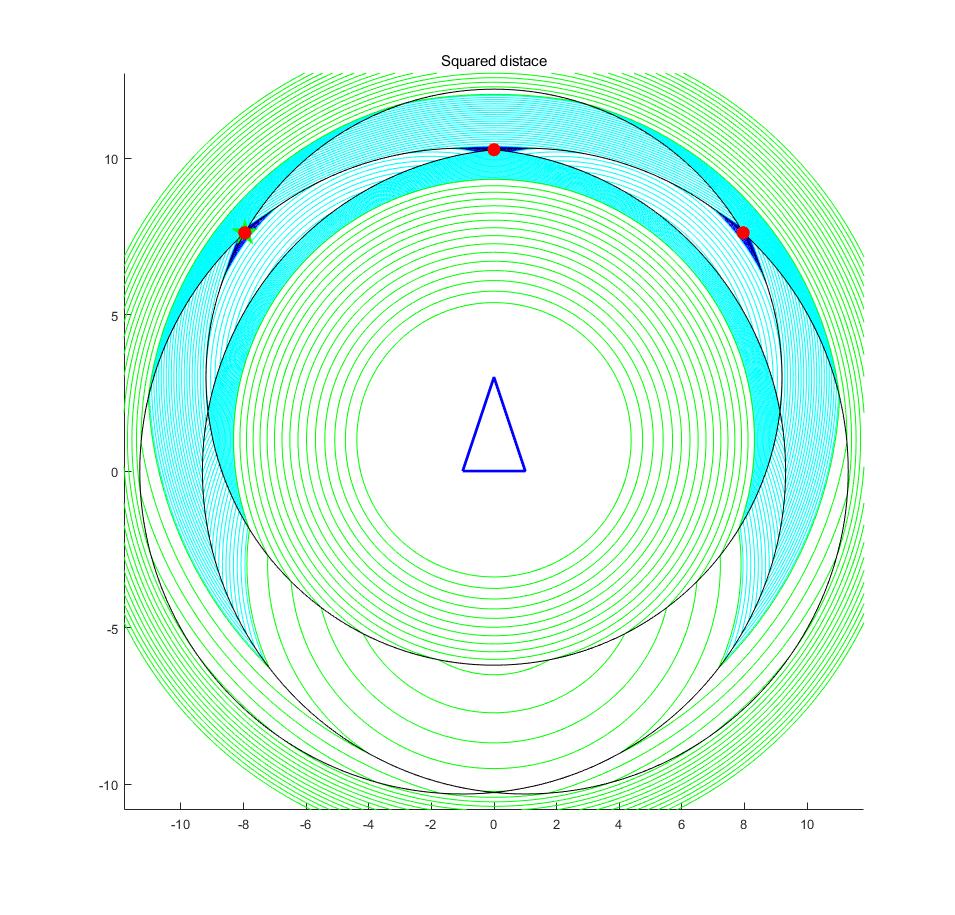, height=4cm,width=4.5cm,clip=1cm}}
\end{minipage}
\begin{center}
(g)\qquad\quad\qquad\qquad\qquad\qquad\qquad\qquad\qquad\qquad\qquad(h)
\end{center}
\caption{The location of the sources for $\psi(y)=y^2$ with $d_1=d_2$ are as follows: For the equilateral triangel cases: (a,c) 3 solutions (b) 4 solutions. 
For flat isosceles cases: (d) 4 solutions (e) 2 solutions. For sharp isosceles cases: (f) 5 solutions (g) 4 solutions (h) 3 solutions.
}
\label{fig:Meas3Dim2_Psi2} 
\end{figure}

In the three dimensional case, we compared the solutions for the distance error ($\psi(y)=y$)  and squred distance error ($\psi(y)=y^2$) scenarios, as illustrated in Fig. \ref{fig:Meas3Dim3}.
Detailed settings for the measurement triangle, the measurement data, and corresponding solutions are provided in Table \ref{tab:Meas3Dim3}.

Three spheres are depicted with  blue, green, and cyan dots. The gradient directions on grid points, with a step size of $1.5$ in all three axis directions, are illustrated using Matlab's `quiver' command. Streamlines are also plotted in these figures. 

\begin{table}
\begin{centering}
\begin{tabular}{|c||c|c|}
\hline
                                                                                  &$\psi(y) =y$                                                                                             &$\psi(y)=y^2$ \\
\hline \hline
$S_{123}\neq\phi$                                                   &$X=S_{123}$ (Fig.\ref{fig:Meas3Dim3}(a))                                          &$X=S_{123}$(Fig.\ref{fig:Meas3Dim3}(b))\\
\hline
$|T|=2$                                                                      &$X=T$(Fig.\ref{fig:Meas3Dim3}(c))                                                     &$X=N_1$(Fig.\ref{fig:Meas3Dim3}(d))\\   
\hline
$|T|=1, N_1\in D_2$                                                &$X=T\cap R_0$(Fig.\ref{fig:Meas3Dim3}(e))                                     &$X=S_{12}^+$(Fig.\ref{fig:Meas3Dim3}(f))\\
\hline 
$Y_{100}\cap R_{100}\neq\phi$ and $\triangle$ is equilateral
                                                                                  &$X=Y_{100}\cap R_{100}, |X|=\infty$(Fig.\ref{fig:Meas3Dim3}(g))    &$X=Y_{010}, |X|=1$(Fig.\ref{fig:Meas3Dim3}(h))\\
\hline                                                                                                                               
\end{tabular}
\caption{Comparison for the solutions for $\psi(y)$ and $\psi(y)=y^2$ in three dimensions. }
\label{tab:Meas3Dim3}
\end{centering}
\end{table}

\begin{figure}
\begin{minipage}[t]{8cm}
\centerline{\epsfig{file=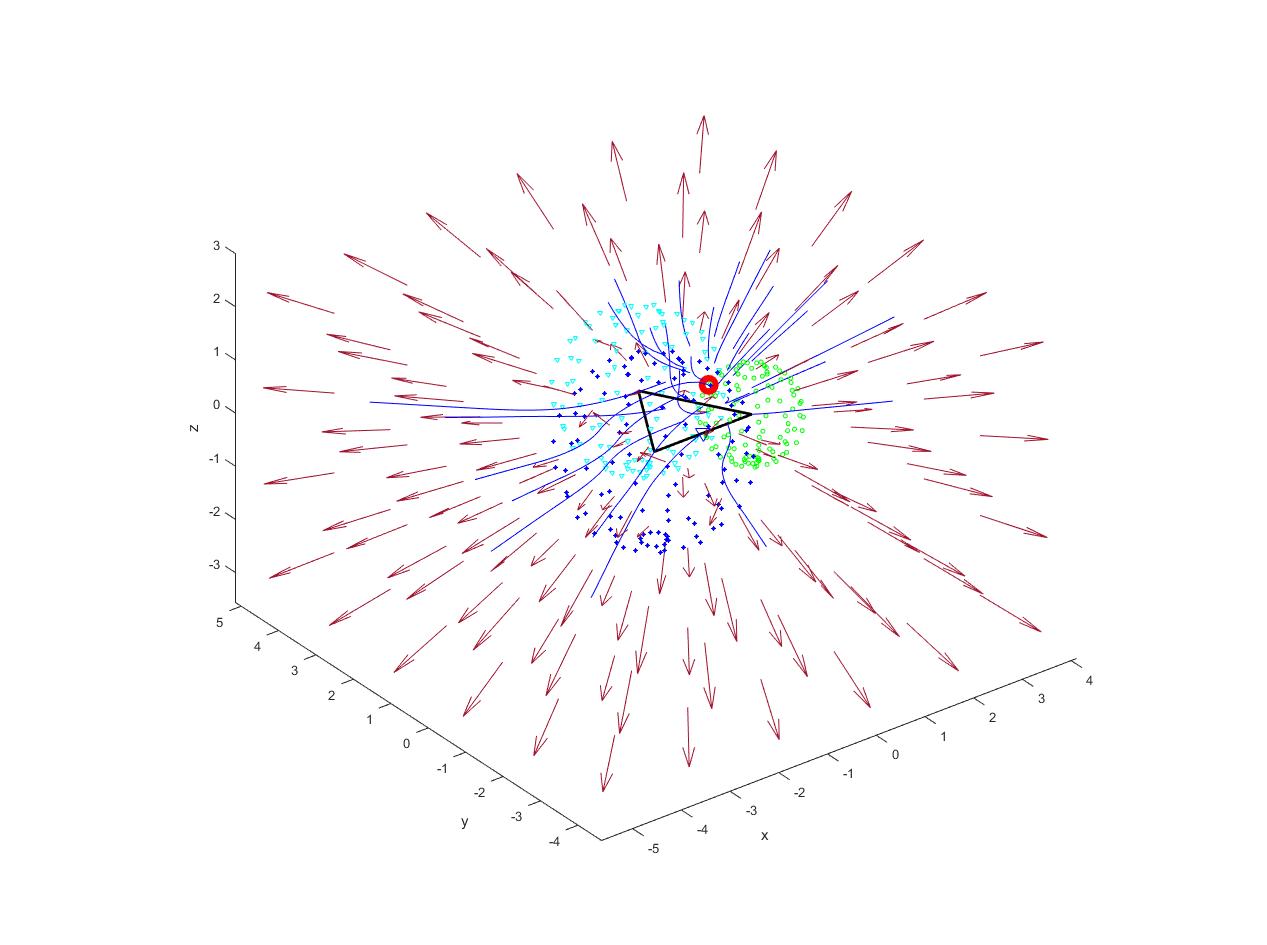, height=4cm,width=5.5cm,clip=1cm}}
\end{minipage}
\begin{minipage}[t]{8cm}
\centerline{\epsfig{file=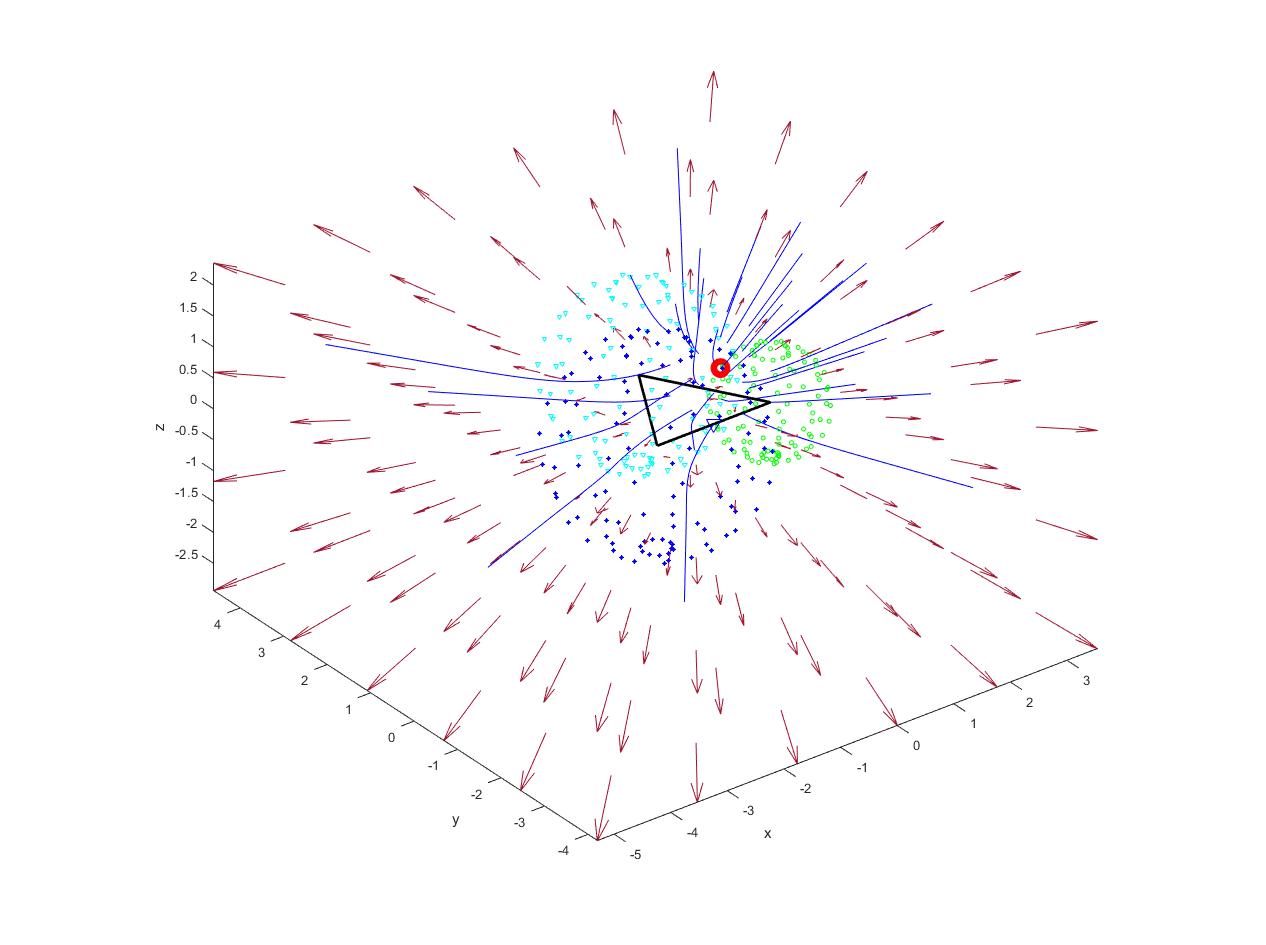, height=4cm,width=5.5cm,clip=1cm}}
\end{minipage}
\begin{center}
(a)\qquad\quad\qquad\qquad\qquad\qquad\qquad\qquad\qquad\qquad\qquad(b)
\end{center}
\begin{minipage}[t]{8cm}
\centerline{\epsfig{file=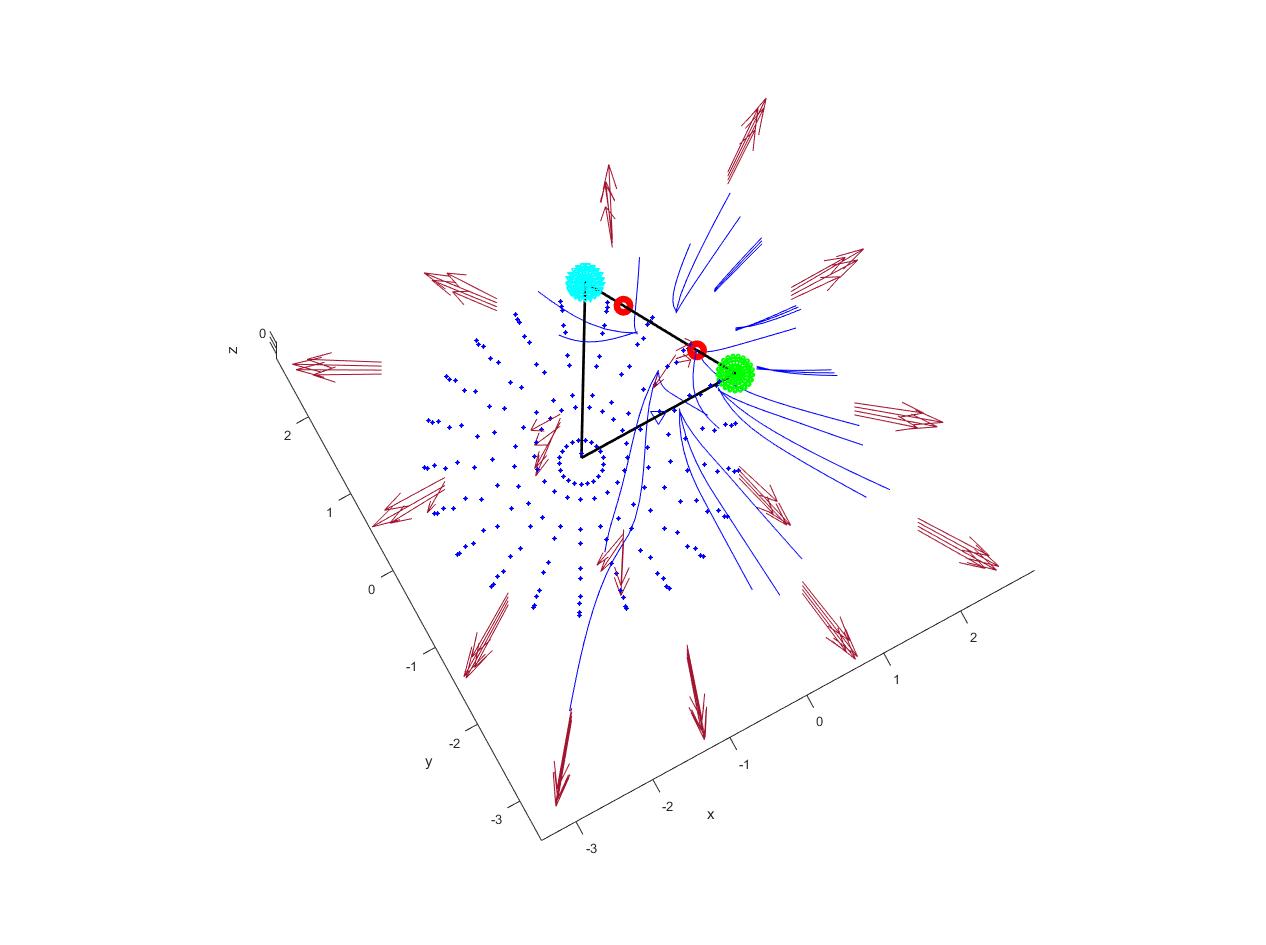, height=4cm,width=5.5cm,clip=1cm}}
\end{minipage}
\begin{minipage}[t]{8cm}
\centerline{\epsfig{file=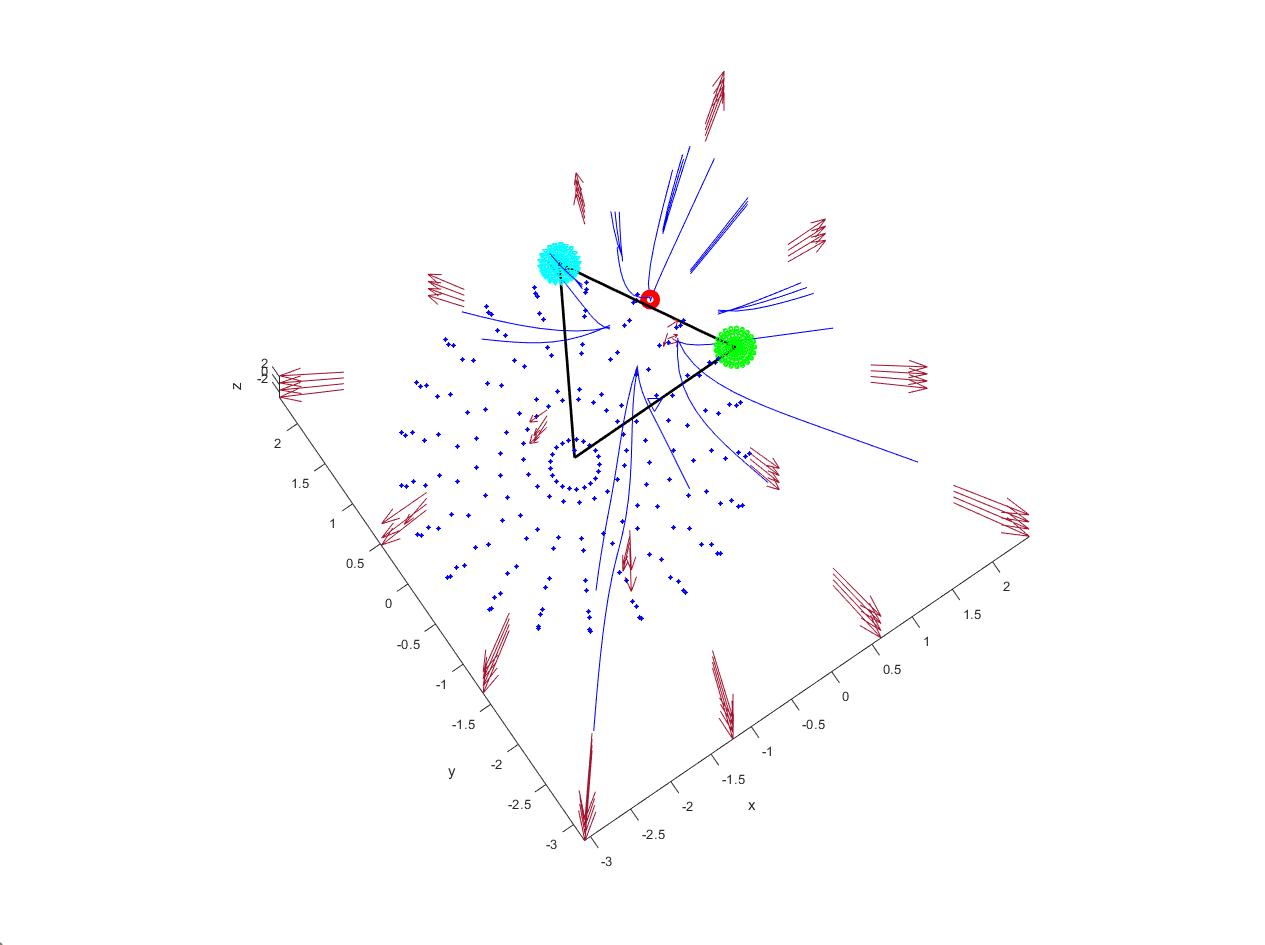, height=4cm,width=4.5cm,clip=1cm}}
\end{minipage}
\begin{center}
(c)\qquad\quad\qquad\qquad\qquad\qquad\qquad\qquad\qquad\qquad\qquad(d)
\end{center}
\begin{minipage}[t]{8cm}
\centerline{\epsfig{file=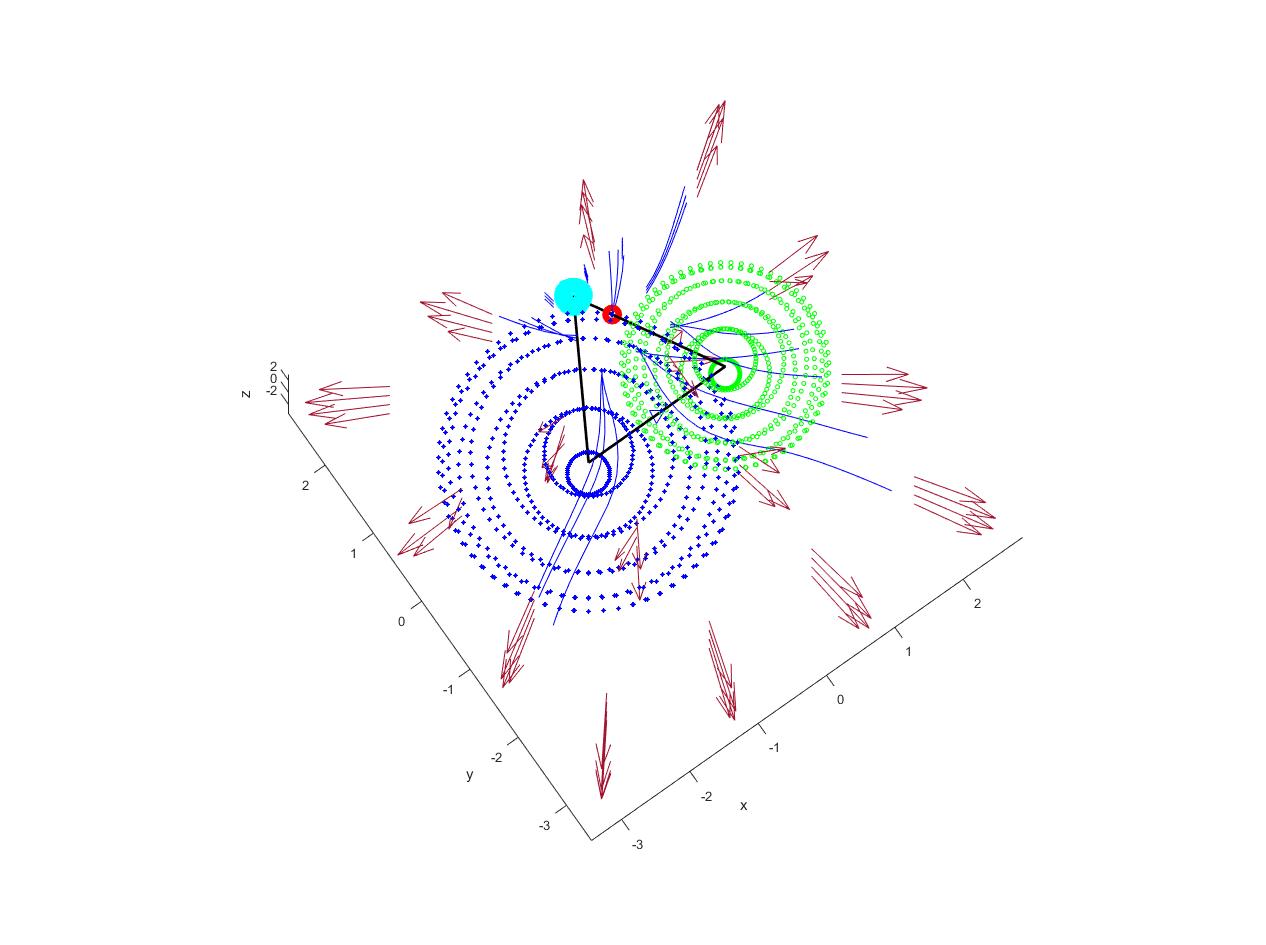, height=4cm,width=5.5cm,clip=1cm}}
\end{minipage}
\begin{minipage}[t]{8cm}
\centerline{\epsfig{file=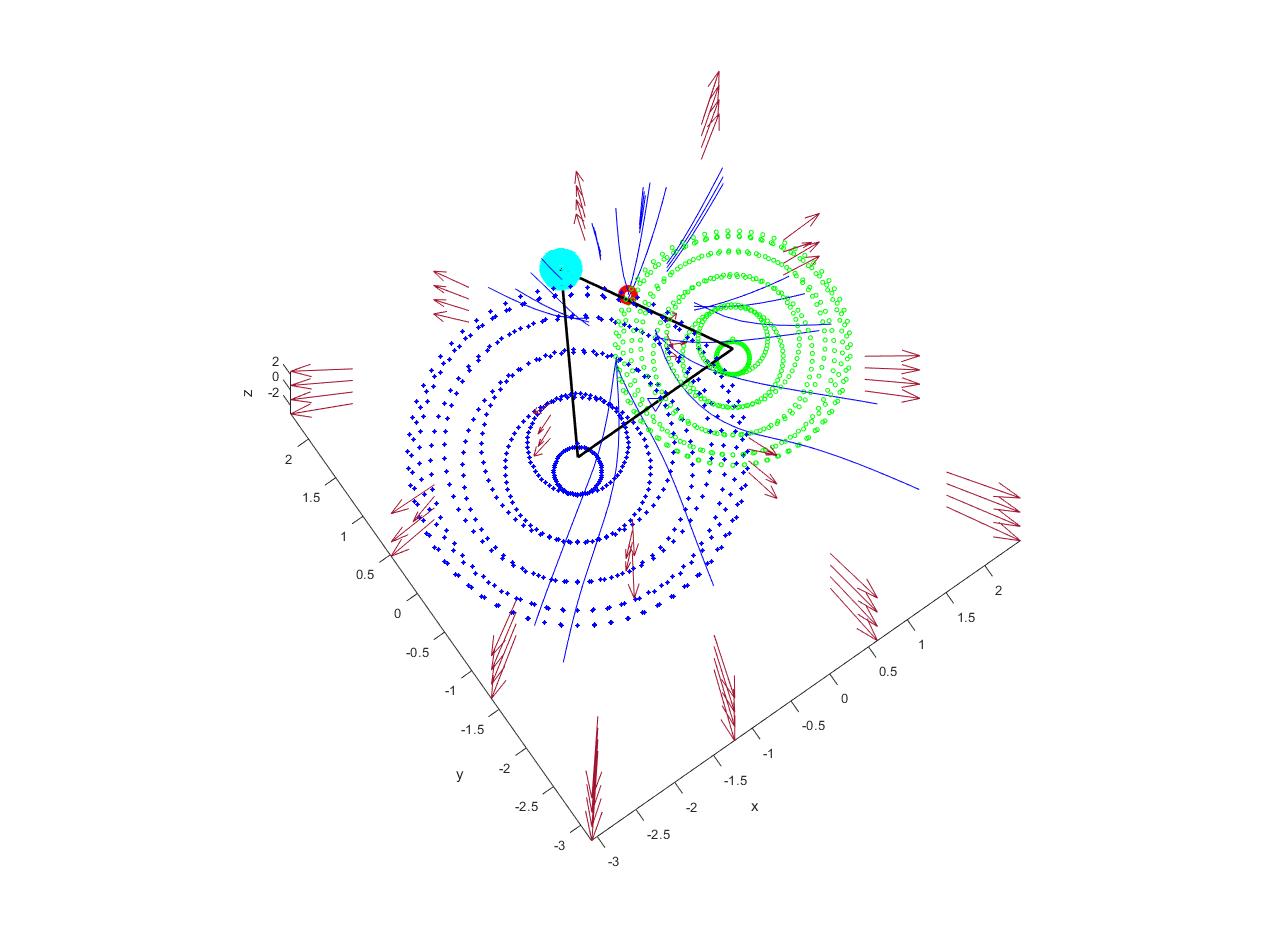, height=4cm,width=5.5cm,clip=1cm}}
\end{minipage}
\begin{center}
(e)\qquad\quad\qquad\qquad\qquad\qquad\qquad\qquad\qquad\qquad\qquad(f)
\end{center}
\begin{minipage}[t]{8cm}
\centerline{\epsfig{file=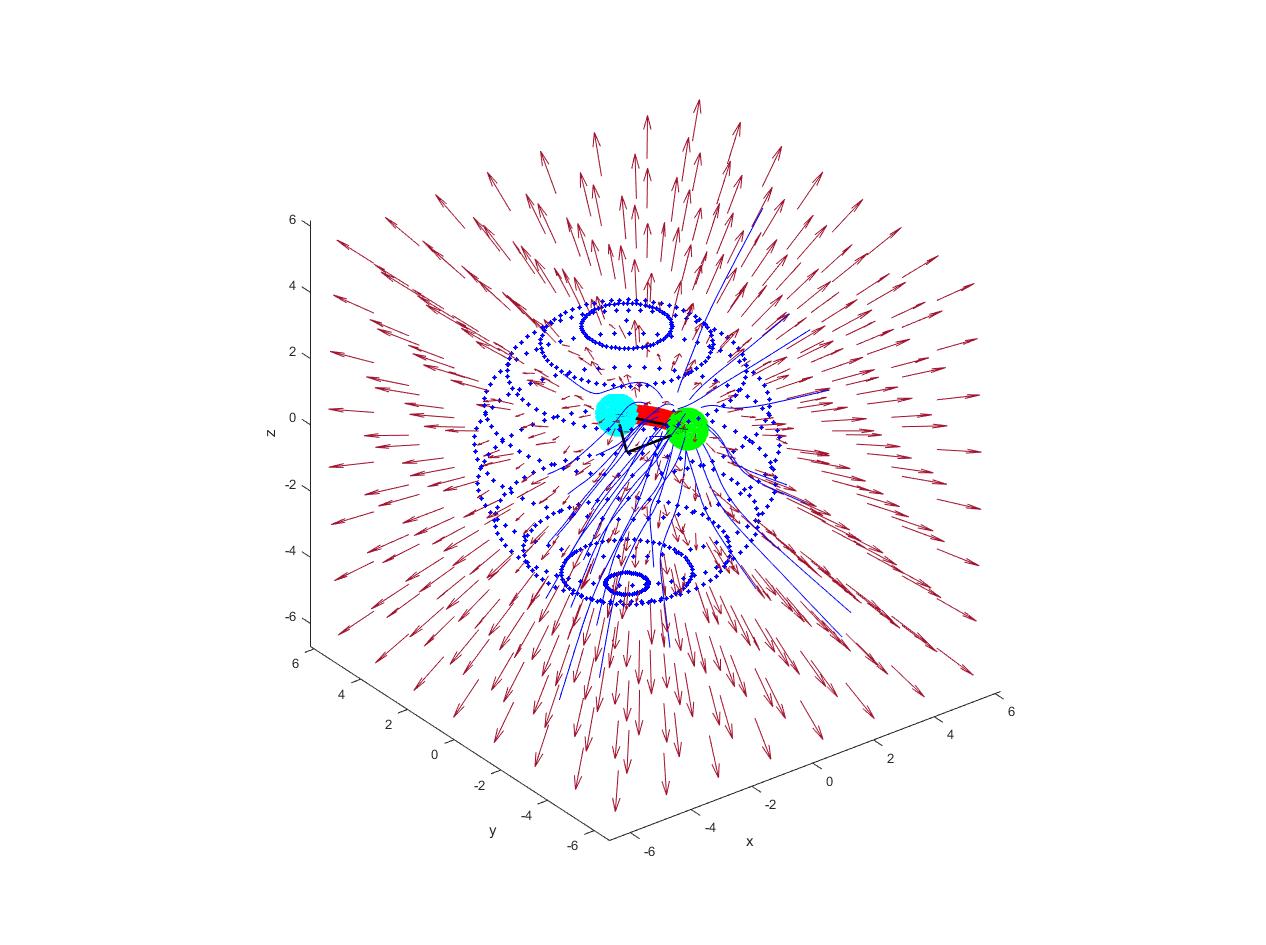, height=4cm,width=5.5cm,clip=1cm}}
\end{minipage}
\begin{minipage}[t]{8cm}
\centerline{\epsfig{file=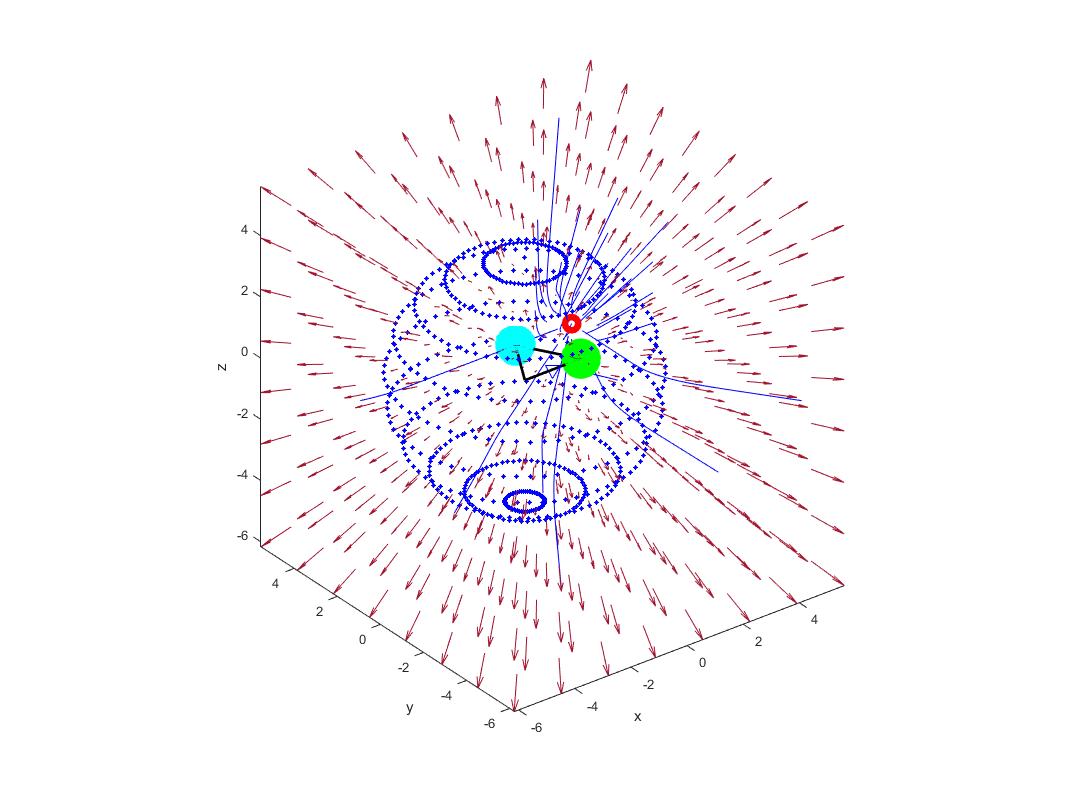, height=4cm,width=5.5cm,clip=1cm}}
\end{minipage}
\begin{center}
(g)\qquad\quad\qquad\qquad\qquad\qquad\qquad\qquad\qquad\qquad\qquad(h)
\end{center}
\caption{Three-dimensional solution with gradient direction quiver plots when (a,b) $S_{123}\neq\phi$, (c,d) $|T|=2$, (e,f) $|T|=1$ and $N_1\in D_2$, and  (g,h) $Y_{100}\cap R_{100}\neq\phi$ with $\triangle$ being equilateral. Plots (a,c,e,g) correspond to the distance error case, and plots (b,d,f,h) correspond to the squared distance error case.  
}
\label{fig:Meas3Dim3} 
\end{figure}

\section{Conclusion}
In this study, we addressed the source localization problem by fomulating two primary minimization problems: distance error and squared distance error minimization. 
We investigated these problems in both three-dimensional space,which is the standard scenario, and in two-dimensional space as a simplified case. 
Additionally, we analyzed the number of possible source solutions when the number of measurements is fewer than three.

Even though a single measurement yields a solution thhat is a disk in 2D and a ball in 3D for both minimizations, the step lengths of the level sets for the distance error cases are more
uniform compared to the squrared distance error cases.

For two measruements, most senarios for the distance error case produced an infinite number of solutions, even in two dimensions. However, in the two-dimensional squared distance error case, 
the maximum number of solutions is 2, while in three-dimensional squared distance error case, an infinite number of solutions occur s in only one of four cases.
This resutt is related to the fact that the intersection of 
two disks is two-point set in two dimension, whereas the intersection of two balls is a circle with infinite points in three dimension, provided the interesection is nonempty nor one-point set and neither ball or disk does not contain the other..

The scenarios for three measurements are not fully explored. The maximum number of solutions is infinite for the distance error case, 5 for two-dimensional squared distance case, 
and 2 for three dimensional squared distance case. The difference  between two- and three- dimensional squared distance cases arises from Lemma \ref{le:3D}.

Extending this research to a multi-measruement analysis will be the focus of future works.  
The existence of multiple potential signal sources adds complexity to the localization process. Accurately distinguishing among these sources and correctly identifying the true one is crucial for ensuring the stability and reliability of the GPS system. This research on source localization with limited measurements aims to resolve potential ambiguities and non-uniqueness in the solutions. A key contribution of this work is the exploration of how these limitations affect the problem and the identification of strategies to mitigate them. 

\section{Acknowledgements}
This paper was supported by the NRF (National Research Foundation of Korea) grant funded by Korean government (Ministry of Science and ICT : RS-2023-00242308).



\begin{thebibliography}{99}                                                                   
%
%
%
%
%
%
%
%
%
%
  
\bibitem{KL1} Kwon K,
\newblock  Uniqueness and nonuniqueness for the $L^1$ minimization source localization problem with three measurements,
\newblock Applied Mathematics and Computation, {\bf 413}(2022)  126649  

\bibitem{KL2} Kwon K,
\newblock  Exact solutions for source localization with squared distance error,
\newblock Applied Mathematics and Computation {\bf 427}(2022) 127187

\bibitem{KL2Cor} Kwon K,
\newblock The number of global solutions for GPS source loacalizaiton in two-dimension, 
\newblock Journal of Mathematics {\bf 2024}(2024) 7980810

\end{thebibliography}
\end{document}